\documentclass[manuscript,screen]{acmart}

\AtBeginDocument{%
  \providecommand\BibTeX{{%
    \normalfont B\kern-0.5em{\scshape i\kern-0.25em b}\kern-0.8em\TeX}}}

\setcopyright{acmlicensed}
\copyrightyear{2024}
\acmYear{2024}
\acmDOI{XXXXXXX.XXXXXXX}

\acmJournal{JACM}
\acmVolume{37}
\acmNumber{4}
\acmArticle{111}
\acmMonth{8}

\acmISBN{978-1-4503-XXXX-X/18/06}

\usepackage[utf8]{inputenc} 
\usepackage[LGR,T1]{fontenc}
\usepackage{hyperref}       
\usepackage{url}            
\usepackage{amsfonts}       
\usepackage{nicefrac}       
\usepackage{microtype}      
\usepackage{subcaption}    
\usepackage{graphicx}
\usepackage{amsmath}
\usepackage{amsthm}
\usepackage[shorthand, inference]{semantic}

\begin{document}
\newcommand{\preset}[2]{\overleftarrow{N}_{#1}(#2)}
\newcommand{\poset}[2]{\overrightarrow{N}_{#1}(#2)}
\newcommand{\fullset}[2]{N_{#1}(#2)}
\newcommand{\inset}[2]{I_{#1}(#2)}
\newcommand{\outset}[2]{O_{#1}(#2)}

\newcommand{\graph}{G}
\newcommand{\vSet}{V}
\newcommand{\eSet}{E}

\newcommand{\gnncomp}{\times}
\newcommand{\relpar}{\frown}
\newcommand{\relcomp}{\ast}
\newcommand{\parcomp}{\otimes}
\newcommand{\seqcomp}{\odot}
\newcommand{\starop}[1]{#1^{\star}}
\newcommand{\choiceop}{\oplus}
\newcommand{\id}{\iota}
\newcommand{\simop}{\simeq_\epsilon}
\newcommand{\mG}{\mu\mathcal{G}}
\newcommand{\len}[1]{|#1|}

\newcommand{\lblfunv}{\eta} 
\newcommand{\lblfunvtype}[1]{H[#1]} 
\newcommand{\lblfunvtypews}{H}
\newcommand{\lblfune}{\xi} 
\newcommand{\lblfunetype}[1]{\Xi[#1]} 

\newcommand{\gnn}{\phi_{\graph,\xi}} 
\newcommand{\gnnws}{\phi} 
\newcommand{\gnntype}[2]{\Phi_{\graph,\xi}[#1, #2]} 
\newcommand{\gnntypews}{\Phi_{\graph,\xi}}
\newcommand{\gnnorder}{\sqsubseteq_{\gnntypews}}
\newcommand{\typedgnnorder}[2]{\sqsubseteq_{\gnntype{#1}{#2}}}

\newcommand{\ds}[1]{\mathcal{S}_{ds}|[#1|]^{\graph,\lblfune}}

\newcommand{\sosf}[1]{\mathcal{S}_{sos}|[#1|]^{\graph,\lblfune}}
\newcommand{\soscomp}[2]{\langle #1 , #2 \rangle}

\newcommand{\sos}[5][]{\soscomp{#2}{#3} \rightarrow_{\graph, \lblfune}^{#1} \soscomp{#4}{#5}}

\newcommand{\sosv}[4][]{\soscomp{#2}{#3} \rightarrow_{\graph, \lblfune}^{#1}  #4} 

\newcommand{\vsos}[4][]{#2 \rightarrow_{\graph, \lblfune}^{#1} \soscomp{#3}{#4}}

\newcommand{\vsosv}[3][]{#2 \rightarrow_{\graph, \lblfune}^{#1} #3}

\title{The \texorpdfstring{$\mG$}{mG} Language for Programming Graph Neural Networks}

\author{Matteo Belenchia}
\orcid{0000-0003-4988-0566}
\affiliation{
  \institution{University of Camerino}
  \city{Camerino}
  \state{Macerata}
  \country{Italy}
}
\email{matteo.belenchia@unicam.it}

\author{Flavio Corradini}
\orcid{0000-0001-6767-2184}
\affiliation{
  \institution{University of Camerino}
  \city{Camerino}
  \state{Macerata}
  \country{Italy}
}
\email{flavio.corradini@unicam.it}

\author{Michela Quadrini}
\orcid{0000-0003-0539-0290}
\affiliation{
  \institution{University of Camerino}
  \city{Camerino}
  \state{Macerata}
  \country{Italy}
}
\email{michela.quadrini@unicam.it}

\author{Michele Loreti}
\orcid{0000-0003-3061-863X}
\affiliation{
  \institution{University of Camerino}
  \city{Camerino}
  \state{Macerata}
  \country{Italy}
}
\email{michele.loreti@unicam.it}

\renewcommand{\shortauthors}{Belenchia et al.}

\begin{abstract}
Graph neural networks form a class of deep learning architectures specifically designed to work with graph-structured data. As such, they share the inherent limitations and problems of deep learning, especially regarding the issues of explainability and trustworthiness. We propose $\mG$, an original domain-specific language for the specification of graph neural networks that aims to overcome these issues. The language's syntax is introduced, and its meaning is rigorously defined by a denotational semantics. An equivalent characterization in the form of an operational semantics is also provided and, together with a type system, is used to prove the type soundness of $\mG$. We show how $\mG$ programs can be represented in a more user-friendly graphical visualization, and provide examples of its generality by showing how it can be used to define some of the most popular graph neural network models. The formal language we developed is well-suited for applying static analysis techniques to formally verify the behavior of graph neural networks.
\end{abstract}



\received{1 October 2024}
\received[revised]{???}
\received[accepted]{???}

\maketitle

\section{Introduction}

Deep learning models are at the forefront of artificial intelligence research today. Among them, artificial neural networks are the most commonly used class of models for a wide range of different tasks, including natural language processing, computer vision, software engineering, and many more~\cite{sarker_deep_2021}. For applications where the inputs can be represented as graphs, the graph neural network~\cite{corso_graph_2024} (GNN) model has been the architecture of choice, and has achieved state-of-the-art performance on many tasks.

Despite these promising advancements, deep learning models, including graph neural networks, face a number of issues. These systems are difficult to engineer with, and are notoriously harder to debug and interpret compared to traditional software systems~\cite{marcus_deep_2019}. Another issue is the lack of guarantees that these systems offer regarding their outputs, which have been shown to be easily foolable, not only using so-called ``adversarial examples''~\cite{sicong_interpreting_2023}, but also more generally in unpredictable and surprising ways~\cite{eykholt_robust_2018,athalye_synthesizing_2018}. Furthermore, these systems act like a black-box and are opaque, in the sense that human users are unable to understand how they have reached their conclusions~\cite{barbierato_challenges_2024}. Finally, there is a very strong bias in deep learning research against the usage of prior knowledge even when such usage is warranted, and even then, it is difficult to figure out how to integrate such knowledge in these systems~\cite{marcus_innateness_2018}.

In this paper, we tackle all these issues by proposing a new graph neural network specification language, called $\mG$ (pronounced as ``mee-gee''). A graph neural network in $\mG$ is built up as a composition of simpler GNNs following the formal rules of the language. These GNNs are built up from the base terms, specifying functions to be applied to the single nodes in terms of themselves and/or their neighbors, which can then be composed sequentially, in parallel, selected according to a Boolean condition, or iterated. 

The language's syntax is introduced using a context-free grammar, and the meaning of each term is formally and rigorously defined by a denotational semantics. Later on, we also introduce a structural operational semantics, and prove that both semantics are equivalent, i.e., they define the same graph neural network operations. Using the operational semantics and after having defined a type system for $\mG$, we prove the type soundness of the language. The type soundness of $\mG$ guarantees that every well-typed $\mG$ program properly defines a GNN with the correct output type.

The language can be used both in textual form and in graphical form. The graphical representation of $\mG$ programs is introduced as a more user-friendly approach to use the language, as it keeps the flow of information and the type of labels easier to follow and reason about.

Our language is framework and hardware-agnostic, as $\mG$ could be implemented in different ways. We opted to implement $\mG$ in TensorFlow so that it can use its automatic differentiation capabilities and allow the GNNs we program to be executable on CPUs, GPUs, or TPUs without changing their implementation. Furthermore, these GNNs can interoperate with any other TensorFlow feature as if they were TensorFlow models themselves.

We claim that $\mG$ helps with the aforementioned problems of deep learning. For the matter of explainability and interpretability, using $\mG$ helps by making the computations performed more explicit, by taking the form of a $\mG$ expression which has clearly definite semantics and types. Indeed, the user might even define functions and computations using the same terminology of the specific domain in which the GNN is going to be used in, making the purpose of each term easier to understand. Some, or all, parts of the GNN can be defined to be inherently interpretable, by virtue of being defined explicitly based on the available domain knowledge. Furthermore, determining the nodes, edges, or labels that contributed to a specific prediction becomes amenable to static analysis techniques~\cite{belenchia_libmg_2024}.

As for the issues of trustworthiness, the use of a formal language like $\mG$ allows the formal verification of GNNs similarly to that of other programming languages, e.g., by abstract interpretation, symbolic execution, data-flow analysis, and so on. In particular, we are interested in the problem of formally verifying safety properties of GNNs, e.g. output reachability properties~\cite{huang_survey_2020}. The verification of an output reachability property for a function $f$ given an input region $X$ and an output region $Y$ consists in checking whether for all inputs $x \in X, f(x) \in Y$. As far as we know, very little work has been done on this problem for graph neural networks specifically~\cite{salzer_fundamental_2022,ladner2024formalverificationgraphconvolutional}. The verification of properties of this kind can be used to prove the robustness of a GNN against adversarial examples, which is critical when such models are deployed in safety-critical systems, where guarantees of correct functioning are paramount.

Finally, it is easy to include prior knowledge when defining a GNN in $\mG$. The basic terms of the language need not make use of neural network layers as is typically the case for the more popular graph neural network models, but can in general use any kind of function, not necessarily depending on trainable parameters. This way is also possible to define a GNN which does not require training at all, as we did in a previous work where we used $\mG$ for model checking~\cite{belenchia_implementing_2023}. Prior, or innate, knowledge can be freely mixed with optimizable function in order to build hybrid, neural-symbolic systems~\cite{marcus_next_2020,lamb_graph_2020,susskind_neuro-symbolic_2021}. 

\paragraph{Contributions} Our main contributions are:
\begin{itemize}
    \item We define the syntax $\mG$ language, and provide its denotational and operational semantics.
    \item We prove the equivalence of the denotational and operational semantics of $\mG$.
    \item We define the typing rules of $\mG$ and prove its type soundness.
    \item We show how $\mG$ programs can be represented graphically.
    \item We demonstrate how to use $\mG$ to program some of the most popular graph neural network models.
\end{itemize}
\paragraph{Structure of the paper} We survey the related work in Section~\ref{sec:literature}, while we introduce the necessary notation and preliminary information in Section~\ref{sec:background} and~\ref{sec:domain}. Then the $\mG$ language is described in detail in Section~\ref{sec:mg} by showing its syntax and semantics. A proof of equivalence of the denotational and operational semantics of $\mG$ is in Section~\ref{sec:equivalence}, followed by a proof of type soundness in Section~\ref{sec:types}. The graphical representation of $\mG$ programs is introduced in Section~\ref{sec:graphics}. A brief discussion on the implementation of $\mG$ is in Section~\ref{sec:implementation}. Finally, we evaluate $\mG$ in Section~\ref{sec:evaluation} by recalling its previous application to CTL model checking and by showing how to define some of the most well-known graph neural network models.

\section{Related Work}\label{sec:literature}
Many domain specific languages (DSL) have been developed over the years to ease the development of machine learning applications~\cite{portugal_preliminary_2016}. Far from being a complete survey, we will discuss here only the most recent ones or those that are most relevant for our work.

DeepDSL~\cite{zhao_deepdsl_2017} is a language based on Scala that allows the definition of neural networks through variable assignments. Variables can be assigned computational modules such as convolutional layers, pooling layers, activation functions and so on. Neural networks are defined by variable assignments where the right-hand side is the function composition of other computational modules. The code is compiled to Java and uses JCuda to communicate with CUDA and cuDNN. AiDSL~\cite{garcia_diaz_towards_2015} is a language for the specification of supervised feed-forward neural networks using a model driven engineering approach. The code is compiled into Encog, a machine learning framework in Java. SEMKIS-DSL~\cite{jahic_semkis-dsl_2023} is another language developed using a model driven engineering approach, where the main focus is shifted from the specification of the neural network architecture to the requirements that the input dataset and the neural network's outputs must satisfy. For the time being there are no compilation targets for the language, so it is not possible to automatically produce a neural network from specification. On the other hand, Diesel~\cite{elango_diesel_2018} acts at a finer level of detail, allowing the specification of common linear algebra and neural network computations that are typically implemented by libraries such as cuBLAS and cuDNN. It uses a polyhedral model to schedule operations and perform optimizations, and the source code is compiled directly into CUDA code. So far, none of these languages take into consideration the definition of graph-structured inputs or the typical operations that characterize graph neural networks.

A DSL where graphs are first-class objects is OptiML~\cite{sujeeth_optiml_2011}, which supports the specification of machine learning models to be run on heterogenous hardware (i.e. the CPU or GPU). The language can be used to implement any machine learning model that can be expressed using the Statistical Query Model~\cite{kearns_efficient_1998}, and the code can be compiled to Scala, C++, or CUDA code. Graph objects in OptiML support operations expressed in terms of nodes and edges, and graph  operations are specified in terms of nodes and their neighbours in the graph. Furthermore, OptiML has a fixed point operator with a customizable threshold value just like $\mG$. Another DSL that supports graphs is StreamBrain~\cite{podobas_streambrain_2021}, which is a language for the specification of Bayesian Confidence Propagation Neural Networks. This model takes in input a graph where every node is a random variable  and edges represent the dependencies between variables. StreamBrain has a Python interface and its syntax is similar to that of Keras, where the model is seen as a stack of layers, and it can compile to OpenMP, OpenCL, and CUDA. Streambrain is based on NumPy and supports CPUs, GPUs and FPGAs. 

Despite the native support for graphs, neither OptiML nor StreamBrain fully support the graph neural network model, with the first being limited to statistical inference problems, and the latter to a specific kind of neural architecture and graph labeling. Therefore, we can conclude that, at the time of writing, no domain specific language has been developed specifically for graph neural network models and tasks.

\section{Background and Notation} \label{sec:background}

\paragraph{Graphs and their labelings.} A directed \textit{graph} is a pair $\graph = (\vSet, \eSet)$ where $\vSet$ is a collection of \textit{vertices} (which we also refer to as \textit{nodes}) and $\eSet$ is a collection of ordered pairs of vertices, called \textit{edges}. 
For any $u,v \in \vSet$, whenever $(u, v)\in \eSet$ we say that $u$ is a \emph{predecessor} of $v$ and that $v$ is a \emph{successor} of $u$. We also say that $(u, v)$ is an \emph{incoming edge} for $v$ and an \emph{outgoing edge} for $u$.
Moreover, let $\preset{\graph}{v}$ denote the set of \emph{predecessors} of $v$, formally $\preset{\graph}{v}=\{ u \mid (u, v)\in \eSet\}$, while $\poset{\graph}{u}$ denotes the set of \emph{successors} of $u$, namely $\poset{\graph}{u}=\{ v \mid (u, v)\in \eSet \}$. Similarly, let $\inset{\graph}{v}$ denote the set of \emph{incoming edges} for $v$, formally $\inset{\graph}{v} = \{ (u, v) \mid u,v \in \vSet \}$, and let $\outset{\graph}{u}$ denote the set of \emph{outgoing edges} for $u$, formally $\outset{\graph}{u} = \{ (u, v) \mid u,v \in \vSet \}$.

Sometimes it is useful to associate nodes and edges with values. Given a graph $\graph =  (\vSet, \eSet)$, we can consider its \emph{node-labeling} and \emph{edge-labeling}. The former is a function $\lblfunv: \vSet \rightarrow T_{\vSet}$ associating each node $v\in \vSet$ with a value in the set $T_{\vSet}$. Similarly, an \emph{edge-labeling} is a function $\lblfune: \eSet \rightarrow T_{\eSet}$ that maps an edge to its label in the set $T_{\eSet}$. 
Let $V\subseteq \vSet$ (resp. $E\subseteq \eSet$), we let $\lblfunv(V)$ (resp. $\lblfune(E)$) denote the \emph{multi-set} of labels associated to the elements of $V$ (resp. $E$) by $\lblfunv$ (resp. $\lblfune$). 


\paragraph{Encoding of graphs and labels} One way to represent graphs and their labeling functions on a computer is in the form of a matrix of node features $\textbf{X}$, an adjacency matrix $\textbf{A}$ and a matrix of edge features $\textbf{E}$. 
The node features matrix $\textbf{X}$ stores the labels associated with each node in the graph. The $i$-th row of $\textbf{X}$ is a value in $T_{\vSet}$ that represents the information associated with the $i$-th node of the graph in a given ordering. The adjacency matrix $\textbf{A}$ encodes the architecture of the graph, with each non-zero element $a_{ij} \in \mathbb{R}$ denoting the presence of an edge from a node $i$ to a node $j$. The edge features matrix $\textbf{E}$ stores the labels associated with each edge, like the node features matrix. The $i$-th row of $\textbf{E}$ is a value in $T_{\eSet}$ that represents the information associated with the $i$-th edge in the row-major (or any other) ordering of the edges in $\textbf{A}$.

\paragraph{Graph Neural Networks} A graph neural network is a deep learning model that operates on graph-structured data. Introduced by~\citet{scarselli_graph_2009}, it emerged to overcome the limitations of graph kernel methods~\cite{morris_power_2021}. GNNs generalize many other classes of deep learning architectures, as other deep learning models can be seen as a particular case of graph neural networks~\cite{bronstein_geometric_2021}. Convolutional Neural Networks (CNNs), for example, can be seen as a graph neural networks where inputs can only be 1-dimensional sequences (such as texts) or 2-dimensional grids (such as images), which are both particular instances of graphs. Graph neural networks, more generally, can learn on any non-Euclidean structure representable as graphs, which contrary to grids such as images, can have differing shapes, number of nodes and connectivity properties.

Graph neural networks can be used for many tasks, which can be roughly categorized according to the subject of prediction: nodes, edges, or entire graphs. As with other machine learning models, there are two main kinds of tasks, namely classification and regression, where the goal is to predict the class or some numerical quantity associated with the nodes in the graph, edges, or the entire graph. 
%
%


%
Over the years, many types of graph neural network variants have been developed~\cite{zhou_graph_2020}. Among these, we mention Graph Convolutional Networks~\cite{kipf_semi-supervised_2017}, Graph Isomorphism Networks~\cite{xu_how_2019}, and Graph Attention Networks~\cite{velickovic_graph_2018}. Graph neural networks have been applied to the most disparate domains, spanning computer vision, natural language processing, particle physics, chemistry, combinatorial optimization, recommender systems, traffic forecasting, graph mining and many more~\cite{zhou_graph_2020}. 


The most general form of graph neural network, which subsumes most of the GNN variants that have been developed over the years~\cite{bronstein_geometric_2021}, is the message passing neural network~\cite{gilmer_neural_2017,battaglia_relational_2018} (MPNN). A typical MPNN comprises a fixed number of convolutional layers stacked in sequence, each of them updating the labels of the nodes. A convolution operation computes the new node labels $x_i'$ for some node $i$ from the current node labels $x_i$ according to the following equation:
    \begin{equation}\label{message-passing-node-level}
        x_i' = \psi\left(x_i, \sigma_{j \in N_{\graph}(i)}\varphi(x_j, e_{ji}, x_i)\right) 
    \end{equation}
where $N_{\graph}(i)$ is the set of neighbors of node $i$ (for any given definition of node neighborhood), $\sigma$ is a permutation-invariant function and $\psi, \varphi$ are (usually) trainable functions such as e.g. neural networks. The function $\varphi$ generates the message sent from node $j$ to node $i$ using the corresponding node labels $x_j, x_i$ and the label $e_{ji}$ of the edge $(j, i)$ connecting them. The function $\sigma$ aggregates the multiset of messages obtained by $\varphi$ in a way that is independent of the order they are evaluated and is usually a non-learnable function such as the sum, product, or mean. The function $\psi$ updates the label of node $i$ by considering both the current label $x_i$ and the value computed by $\sigma$. 


    

Graph neural networks composed by layers described by Equation~\ref{message-passing-node-level} are permutation-equivariant functions that take in input a graph $\graph$, an edge-labeling $\lblfune$, and a node-labeling $\lblfunv$ to return a new node-labeling $\lblfunv'$. 
%
This characterization leaves out some types of graph neural networks, namely, the graph neural networks that use pooling layers (which therefore change the architecture of the graph by removing elements from $\vSet$ and $\eSet$) and the graph neural networks that learn edge features rather than node features (which therefore produce a new edge-labeling $\lblfune'$ instead). Nonetheless, the standard graph neural network model we have just described is general enough to encompass these two special cases as well. A graph neural network that uses pooling layers can be represented by allowing the node-labeling function to label each node with an additional Boolean value that specifies whether the node has been deleted or not and by ruling that an edge is valid if and only if both the source and destination nodes have not been deleted. A graph neural network that learns edge features can instead be modeled by reifying edges into nodes and by having the graph's node-labeling function label each node with an additional Boolean value that specifies whether the node represents an edge or a node of the original graph.

\section{The Domain of Graph Neural Networks}\label{sec:domain}
In this section, we introduce the necessary definitions and theorems that characterize graph neural networks in our work. We elaborate further on the notion of node-labeling function and formalize the graph neural network as a transformation between node-labeling functions, then move on to show that the set of graph neural networks defined this way forms a chain complete partially ordered set (ccpo). This will be later useful when proving the equivalence of the denotational and operational semantics we define in Section~\ref{sec:mg}.

\subsection{Node-labeling functions}
Let $T$ be a non-empty set. We denote the set of node-labeling functions with co-domain $T$ on a graph $\graph$ as $\lblfunvtype{T}$, and we say that this set specifies their \emph{type}. As examples, node-labeling functions of type $\lblfunvtype{\mathbb{B}}$ map nodes to Boolean values, functions of type $\lblfunvtype{\mathbb{N}}$ map nodes to natural numbers, while functions of type $\lblfunvtype{\mathbb{Q}^k}$ with $k \in \mathbb{N}$ map nodes to $k$-tuples of rational numbers. 
%
%

We also want to consider the parallel composition of node-labeling functions. \begin{definition}[Parallel composition of node-labeling functions]
    Given two node-labeling functions $\lblfunv_1: \lblfunvtype{T_1}$ and $\lblfunv_2:\lblfunvtype{T_2}$, their \textit{parallel composition} $\lblfunv_1 | \lblfunv_2$ is defined as the node-labeling function
\[\lambda v . \langle\lblfunv_1(v), \lblfunv_2(v)\rangle\]
with type $\lblfunvtype{T_1 \times T_2}$ that maps nodes to (possibly nested) pairs of node labels.
\end{definition}
\noindent The inverse operation is given by the projection functions $\pi_L, \pi_R: \lblfunvtype{T \times T} \rightarrow \lblfunvtype{T}$ such that
\[
\pi_L(\lblfunv_1 | \lblfunv_2) = \lblfunv_1
\]
\[
\pi_R(\lblfunv_1 | \lblfunv_2) = \lblfunv_2
\]
We call $\pi_L$ the \textit{left projection} and $\pi_R$ the \textit{right projection}. Adequate composition of these two projection functions can be used to obtain any nested node-labeling function obtained through repeated parallel composition. 

\subsection{Graph neural networks}

Given a graph $\graph$ together with its edge-labeling function $\lblfune$, we define a graph neural network to be a partial function $\gnn: \lblfunvtype{T_1} \rightarrow \lblfunvtype{T_2}$ that maps an input node-labeling function to an output node-labeling function. We denote the set of such graph neural networks as $\gnntype{T_1}{T_2}$ or, more succinctly, as $\gnntypews$. The subscript $\graph, \lblfune$ indicates that each graph neural network is parametrized by a graph and an edge-labeling function. Then, the set of graph neural networks $\gnntype{T_1}{T_2}$ is a partially ordered set, where
$\typedgnnorder{T_1}{T_2}$ is the point-wise ordering relation such that, for all $\lblfunv \in \lblfunvtype{T_1}$
\[\gnnws_1 \typedgnnorder{T_1}{T_2} \gnnws_2 \iff \gnnws_1(\lblfunv) = \lblfunv' \implies \gnnws_2(\lblfunv) = \lblfunv'\]
When the types are clear from the context or not relevant, we drop the subscript and simply write $\gnnorder$.

Next, we define the underlying \textit{relation} of a GNN as the set of input-output tuples which characterize the GNN. Then, building on this concept, we specify the \textit{sequential} and \textit{parallel composition} of relations.
\begin{definition}
    Given a graph neural network $\gnnws: \gnntype{T_1}{T_2}$ we define its underlying \emph{relation}, denoted by $rel(\gnnws)$, as 
    \[rel(\gnnws) = \{ (\lblfunv_1, \lblfunv_2) \in \lblfunvtype{T_1} \times \lblfunvtype{T_2} \mid \gnnws(\lblfunv_1) = \lblfunv_2 \}\]
\end{definition}
\begin{definition}
    The \emph{sequential composition}, or simply \emph{composition}, of relations $A: \lblfunvtype{T_1} \times \lblfunvtype{T_2}$ and $B: \lblfunvtype{T_2} \times \lblfunvtype{T_3}$, denoted by $A \relcomp B$, is defined as
    \[A \relcomp B = \{ (\lblfunv_1, \lblfunv_3) \mid \exists \lblfunv_2 \in \lblfunvtype{T_2}: (\lblfunv_1, \lblfunv_2) \in A \land (\lblfunv_2, \lblfunv_3) \in B\}\]
\end{definition}
\begin{definition}
    The \emph{parallel composition}, or \emph{concatenation}, of relations $A: \lblfunvtype{T_1} \times \lblfunvtype{T_2}$ and $B: \lblfunvtype{T_1} \times \lblfunvtype{T_3}$, denoted by $A \relpar B$, is defined as
    \[A \relpar B = \{ (\lblfunv_1, (\lblfunv_2, \lblfunv_3)) \mid (\lblfunv_1, \lblfunv_2) \in A \land (\lblfunv_1, \lblfunv_3) \in B\}\]
\end{definition}
The following lemma defines the relationship between the relations and the ordering of GNNs.
\begin{lemma}    
    Given two GNNs $\gnnws_1, \gnnws_2$ we have \[\gnnws_1  \gnnorder \gnnws_2 \iff rel(\gnnws_1) \subseteq rel(\gnnws_2)\]
\end{lemma}
\begin{proof}
    First we show that $\gnnws_1 \gnnorder \gnnws_2 \implies rel(\gnnws_1) \subseteq rel(\gnnws_2)$. Since $\gnnws_1(\lblfunv) = \lblfunv' \implies \gnnws_2(\lblfunv) = \lblfunv'$,  any $(\lblfunv, \lblfunv') \in rel(\gnnws_1)$ is also a member of $rel(\gnnws_2)$.

    From the other direction, suppose that $rel(\gnnws_1) \subseteq rel(\gnnws_2)$. Since $(\lblfunv, \lblfunv') \in rel(\gnnws_1) \implies (\lblfunv, \lblfunv') \in rel(\gnnws_2)$, we have that $\forall(\lblfunv, \lblfunv') \in rel(\gnnws_1), \gnnws_1(\lblfunv) = \lblfunv' \implies \gnnws_2(\lblfunv) = \lblfunv'$.
\end{proof}
As was the case for the node-labeling functions, we consider the parallel composition of graph neural networks. 
\begin{definition}[Parallel composition of graph neural networks]\label{def:par-gnn}
    Given two graph neural networks $\gnnws_1: \gnntype{T}{T_1}$ and $\gnnws_2: \gnntype{T}{T_2}$, their \textit{parallel composition} $\gnnws_1 \gnncomp \gnnws_2$ is defined as the graph neural network 
\[\lambda e . \lambda v . \langle\gnnws_1(e)(v), \gnnws_2(e)(v)\rangle\]
with type $\gnntype{T}{T_1 \times T_2}$ that concatenates the node-labeling functions generated by $\gnnws_1$ and $\gnnws_2$.
\end{definition} 
The relationship between the parallel composition of labeling functions and the parallel composition of graph neural networks is highlighted by the following equation:
\[\gnnws_1 \gnncomp \gnnws_2 (\lblfunv) = \gnnws_1(\lblfunv) | \gnnws_2(\lblfunv)\]

Finally, we prove that the partially ordered set of graph neural networks of a given type $\gnntypews$ with the ordering relation $\gnnorder$ is a chain complete partially ordered set.

\begin{theorem}
The set of graph neural networks $\gnntypews$ is a chain complete partially ordered set $(\gnntypews, \gnnorder)$ with bottom element $\bot_{\gnntypews} = \lambda \lblfunv . \mathtt{undef}$ the GNN that is undefined for every node-labeling function.
\end{theorem}

\begin{proof}
    First, we show that $\gnnorder$ fulfills the requirements to a partial order. We have that $\gnnws \gnnorder \gnnws$ because $\gnnws(\lblfunv) = \lblfunv'$ implies $\gnnws(\lblfunv) = \lblfunv'$. Then to prove transitivity, we assume that $\gnnws_1 \gnnorder \gnnws_2$ and $\gnnws_2 \gnnorder \gnnws_3$ and show that $\gnnws_1 \gnnorder \gnnws_3$. If $\gnnws_1(\lblfunv) = \lblfunv'$, we get that $\gnnws_2(\lblfunv) = \lblfunv'$ from $\gnnws_1 \gnnorder \gnnws_2$. Then we also have $\gnnws_3(\lblfunv) = \lblfunv'$ from $\gnnws_2 \gnnorder \gnnws_3$. Lastly, to prove that the ordering is anti-symmetric, we assume that $\gnnws_1 \gnnorder \gnnws_2$ and $\gnnws_2 \gnnorder \gnnws_1$ and we need to show that $\gnnws_1 = \gnnws_2$. If $\gnnws_1(\lblfunv) = \lblfunv'$, then by $\gnnws_1 \gnnorder \gnnws_2$ we have that $\gnnws_2(\lblfunv) = \lblfunv'$. Likewise if $\gnnws_1(\lblfunv)$ is undefined, then $\gnnws_2(\lblfunv)$ must be undefined as well, otherwise $\gnnws_2(\lblfunv) = \lblfunv'$ and  $\gnnws_2 \gnnorder \gnnws_1$ are in contradiction. Therefore $\gnnws_1$ and $\gnnws_2$ are equal on any $\lblfunv$.

    To see that $(\gnntypews, \gnnorder)$ is a ccpo, we need to show that for all chains $Y$, the least upper bound $\bigsqcup Y$ exists. Let $rel(\bigsqcup Y) = \bigcup \{rel(\gnnws) \mid \gnnws \in Y\}$. We first show that $ \bigcup \{rel(\gnnws) \mid \gnnws \in Y\}$ indeed specifies a graph neural network. That is, whenever $(\lblfunv, \lblfunv')$ and $(\lblfunv, \lblfunv'')$ are members of $X = \bigcup \{rel(\gnnws) \mid \gnnws \in Y\}$, then $\lblfunv' = \lblfunv''$. If $(\lblfunv, \lblfunv') \in X$, then there must be a $\gnnws \in Y$ such that $\gnnws(\lblfunv) = \lblfunv'$, and similarly for $(\lblfunv, \lblfunv'') \in X$, there must be a $\gnnws' \in Y$ such that $\gnnws'(\lblfunv) = \lblfunv''$. Since $Y$ is a chain, then either $\gnnws \gnnorder \gnnws'$ or $\gnnws' \gnnorder \gnnws$. In any case, this means that $\gnnws(\lblfunv) = \gnnws'(\lblfunv)$ and $\lblfunv' = \lblfunv''$. Next we show that $\bigsqcup Y$ as we defined it is an upper bound of $Y$. Let $\gnnws$ be a member of $Y$. Clearly, $\gnnws \gnnorder \bigsqcup Y$, because $rel(\gnnws) \subseteq rel(\bigsqcup Y)$. Lastly, we prove that $\bigsqcup Y$ is the least upper bound of $Y$. Let $\gnnws_1$ be an upper bound of $Y$. By definition, $\gnnws \gnnorder \gnnws_1$ for all $\gnnws \in Y$, and $rel(\gnnws) \subseteq rel(\gnnws_1)$. Then it must also be the case that $\bigcup \{rel(\gnnws) \mid \gnnws \in Y\} \subseteq rel(\gnnws_1)$, and therefore $rel(\bigsqcup Y) \subseteq rel(\gnnws_1)$. Then $\bigsqcup Y \gnnorder \gnnws_1$ and it is the least upper bound of $Y$.

    The last step of the proof is to show that $\bot_{\gnntypews}$ is the least element of $\gnntypews$. It is indeed a member of $\gnntypews$ and $\bot_{\gnntypews} \gnnorder \gnnws$ for all $\gnnws$, since $\bot_{\gnntypews}(\lblfunv) = \lblfunv'$ implies (vacuously) that $\gnnws(\lblfunv) = \lblfunv'$.
\end{proof}


    


    

\section{The \texorpdfstring{$\mG$}{mG} Language for Graph Neural Networks}\label{sec:mg}

In this section, we introduce the syntax of $\mG$ as a programming language for the definition of graph neural networks. After specifying its syntax (Definition~\ref{def:syntax}), we show its denotational semantics (Definition~\ref{def:denotational_semantics}), and structural operational semantics (Definition \ref{def:sos}). 

\begin{definition}[Syntax of $\mG$]\label{def:syntax}
Given a set $\mathcal{S}$ of function symbols, we define an algebra of graph neural networks with the following abstract syntax:
\begin{normalfont}
\begin{align*} 
\mathcal{N} ::=& \; \id \mid \psi \mid \lhd_{\sigma}^{\varphi} \mid \rhd_{\sigma}^{\varphi} \mid \mathcal{N}_1 ; \mathcal{N}_2 \mid \mathcal{N}_1 || \mathcal{N}_2 \mid 
\mathcal{N}_1 \choiceop \mathcal{N}_2 \mid  \starop{\mathcal{N}} 
\end{align*}
\end{normalfont}
\noindent with $\varphi, \sigma, \psi \in \mathcal{S}$. The operator precedence rules given by $\starop{} > \, ; \, > || > \choiceop$ and round parentheses are introduced to the syntax to make the meaning of expressions unambiguous.
\end{definition}

Given a graph $\graph$ and an edge-labeling $\lblfune$, the meaning of a $\mG$ expression is a graph neural network, a partial function between node-labeling functions.
The term $\id$ represents the application of the \textit{identity} GNN that leaves the node labels unaltered.
Another of the basic $\mG$ terms is the \emph{function application} $\psi$. This represents the GNN that applies the function associated with $\psi$.
Moreover, the \emph{pre-image} term $\lhd_{\sigma}^{\varphi}$ and the \emph{post-image} term $\rhd_{\sigma}^{\varphi}$ define a GNN that computes the labeling of a node in terms of the labels of its predecessors and successors, respectively, using functions associated with symbols $\sigma$ and $\varphi$.
Two GNNs can be composed by \emph{sequential composition} $\mathcal{N}_1 ; \mathcal{N}_2$ and \emph{parallel composition} $\mathcal{N}_1 || \mathcal{N}_2$. 
%
%
The \textit{choice} operator $\mathcal{N}_1 \choiceop \mathcal{N}_2$ allows to run different GNNs according to the values of a node-labeling function.
Finally, the \emph{star} operator $\starop{\mathcal{N}}$ is used to program recursive behavior. 
%

\subsection{Denotational semantics}
Having defined the syntax, we are now ready to introduce the denotational semantics of $\mG$.

\begin{definition}{(Denotational semantics)} \label{def:denotational_semantics}
Given a graph $\graph$ and an edge-labeling function $\lblfune$, we define the semantic interpretation function $\ds{\cdot}: \mathcal{N} \rightarrow \gnntypews$ on $\mu\mathcal{G}$ formulas $\mathcal{N}$ by induction in the following way:
\begin{normalfont}
\begin{align*}
&\ds{\id} = id \\
&\ds{\psi}(\lblfunv)  = \lambda v . f_{\psi}(\lblfunv(\vSet), \lblfunv(v)) \\
&\ds{\lhd_{\sigma}^{\varphi}}(\lblfunv) =  \lambda v . f_\sigma([f_\varphi(\lblfunv(u), \lblfune((u, v)), \lblfunv(v)) \mid (u, v) \in \inset{\graph}{v}], \lblfunv(v)) \\
&\ds{\rhd_{\sigma}^{\varphi}}(\lblfunv) =  \lambda v . f_\sigma([f_\varphi(\lblfunv(u), \lblfune((v, u)), \lblfunv(v)) \mid (v, u) \in \outset{\graph}{v}], 
\lblfunv(v)) \\
&\ds{\mathcal{N}_1 ; \mathcal{N}_2}  = \ds{\mathcal{N}_2} \circ \ds{\mathcal{N}_1} \\
&\ds{\mathcal{N}_1 || \mathcal{N}_2} = \ds{\mathcal{N}_1} \gnncomp \ds{\mathcal{N}_2}\\
&\ds{\mathcal{N}_1 \choiceop \mathcal{N}_2} = cond(\pi_L, \ds{\mathcal{N}_1} \circ \pi_R, \ds{\mathcal{N}_2} \circ \pi_R)\\
&\ds{\starop{\mathcal{N}}} = FIX(\lambda \gnnws .  cond(\lambda e . \lambda v . \ds{\mathcal{N}}(e)(v) \simop e(v), id, \gnnws \circ \ds{\mathcal{N}}))
\end{align*}
\end{normalfont}
for any $\psi, \varphi, \sigma \in \mathcal{S}$ and any $f_\psi:T_1^\star \times T_1 \rightarrow T_2, f_\varphi: T_1 \times T_e \times T_1 \rightarrow T_2, f_\sigma: T_2^\star \times T_1 \rightarrow T_3$. The functions $cond$ and $FIX$ are defined as:

\[
cond(t, f_{1}, f_{2})(\lblfunv) = \begin{cases} 
f_{1}(\lblfunv) & \text{if } t(\lblfunv) = \lambda v . \mathtt{True} \\
f_{2}(\lblfunv) & \text{otherwise}
\end{cases}
\]

\[FIX(f) = \bigsqcup \{f^n(\bot_{\gnntypews}) \mid n \geq 0 \}\]
where $f^0 = id$ and $f^{n+1} = f \circ f^n$.
For any label type $T$ and any rational value $\epsilon \in \mathbb{Q}$, we define a binary predicate $\simop$ such that $\forall x, y \in T$
\[x \simop y \iff (x = y) \lor (\exists k \in \mathbb{N}. T \subseteq \mathbb{Q}^k \land \forall i \in \{1,\ldots, k\} . |x_i - y_i| \leq \epsilon) \]
This predicate specifies that for two labels to be considered the ``same'' they must either be exactly equal, or they must consist of numerical values and each of those values must not differ by more than $\epsilon$ at most.

%
%
Finally, we provide the semantics for a term of the form $\mathcal{N}_1 \parcomp \mathcal{N}_2$ which will be useful later in the proof of the equivalence with the structural operational semantics
\begin{normalfont}
\begin{align*}
&\ds{\mathcal{N}_1 \parcomp \mathcal{N}_2} = (\ds{\mathcal{N}_1} \circ \pi_L) \gnncomp (\ds{\mathcal{N}_2} \circ \pi_R)
\end{align*}
\end{normalfont}
\end{definition}

\noindent In the following paragraphs, we clarify the meaning of $\mG$ expressions.

\paragraph{Identity application} The term $\id$ is evaluated as the identity graph neural network that returns the input node-labeling function as is.

\paragraph{Function application} A function symbol $\psi \in \mathcal{S}$ is evaluated as the graph neural network that maps a node-labeling to a new node-labeling by applying the corresponding function $f_\psi: T_1^\star \times T_1 \rightarrow T_2$ on both local and global node information. The local information is the label of each individual node, while the global information is the multiset of the labels of all the nodes in the graph. The graph neural network we obtain applies a (possibly trainable) function $f_\psi$ to these two pieces of information. Two particular cases arise if the function ignores either of the two inputs. If $f_\psi$ ignores the global information, the GNN returns a node-labeling function that is a purely local transformation of the node labels. On the other hand, if $f_\psi$ ignores the local information, the GNN returns a node-labeling function that assigns to each node a label that summarizes the entire graph, emulating what in the GNN literature is known as a \textit{global pooling} operator~\cite{grattarola_understanding_2022}.

\paragraph{Pre-image and Post-Image} The pre-image $\lhd$ and the post-image $\rhd$, together with function symbols $\varphi, \sigma \in \mathcal{S}$ are evaluated as the graph neural networks $\ds{\lhd_{\sigma}^{\varphi}}$ and $\ds{\rhd_{\sigma}^{\varphi}}$. In the case of the pre-image, for any symbol $\varphi \in \mathcal{S}$ the corresponding function $f_\varphi: T_1 \times T_e \times T_1 \rightarrow T_2$ generates a \emph{message} from tuples $(\lblfunv(u), \lblfune((u, v)), \lblfunv(v))$ for each $(u, v) \in \inset{\graph}{v}$. Then for any symbol $\sigma \in \mathcal{S}$ the corresponding function $f_\sigma: T_2^\star \times T_1 \rightarrow T_3$ generates a new label for a node $v$ from the multiset of incoming messages for $v$ obtained from $f_\varphi$ and the current label $\lblfunv(v)$. The functions $f_\varphi$ and $f_\sigma$ may be trainable. The case of the post-image is analogous, with the difference that $f_\varphi$ is applied to tuples $(\lblfunv(v), \lblfune((v, u)), \lblfunv(u))$ for each $(v, u) \in \outset{\graph}{v}$ instead.

\paragraph{Sequential composition} An expression of the form $\mathcal{N}_1 ; \mathcal{N}_2$ is evaluated as the graph neural network resulting from the function composition of $\ds{\mathcal{N}_2}$ and $\ds{\mathcal{N}_1}$.

\paragraph{Parallel composition} An expression of the form $\mathcal{N}_1 || \mathcal{N}_2$ is evaluated as the graph neural network resulting from the parallel composition of $\ds{\mathcal{N}_1}$ and $\ds{\mathcal{N}_2}$ (see Definition~\ref{def:par-gnn}). 

\paragraph{Choice}
The \emph{choice} operator $\mathcal{N}_1 \choiceop \mathcal{N}_2 $ applied to $\mG$ formulas $\mathcal{N}_1, \mathcal{N}_2$ is evaluated as the graph neural network $\ds{\mathcal{N}_1}$ if the left projection of the input node-labeling $\lblfunv' = \pi_L(\lblfunv)$ is a node-labeling function such that $\forall v \in \graph, \lblfunv'(v) = \mathtt{True}$. Otherwise, it is evaluated as the graph neural network $\ds{\mathcal{N}_2}$. The GNN selected to run receives in input the right projection of the input node-labeling function.

\paragraph{Fixed points} The \textit{star} operator $\starop{\mathcal{N}}$, or the \emph{fixed point} operator, applied to a $\mG$ formula $\mathcal{N}$ is evaluated as the graph neural network that that maps a node-labeling function $\lblfunv$ to a new node-labeling function $\lblfunv'$ that is the fixed point of $\mathcal{N}$ computed starting by $\lblfunv$. In other words, the sequence
\begin{align*}
    \lblfunv_0 &= \lblfunv\\
    \lblfunv_1 &= \ds{\mathcal{N}}(\lblfunv_0)\\
    \lblfunv_2 &= \ds{\mathcal{N}}(\lblfunv_1)\\
    & \qquad \vdots\\
    \lblfunv_i &= \ds{\mathcal{N}}(\lblfunv_{i-1})
\end{align*}
is computed until we obtain a labeling $\lblfunv'$ such that $\ds{\mathcal{N}}(\lblfunv') = \lblfunv'$. The $FIX$ function used in the definition requires that the input functional is a continuous function. In order to prove that, we first have to show that function composition $\circ$ and the function $cond$ are continuous. Later, in Section~\ref{sec:equivalence}, we will also require the continuity of parallel composition $\gnncomp$ to prove the equivalence of the denotational and operational semantics. The following lemmas prove all these results.

\begin{lemma}\label{lemma:cont-circ}
    Function composition $\circ$ is continuous in both its arguments.
\end{lemma}

\begin{proof}
We prove that $\circ$ is continuous in the first argument. The proof of continuity in the second argument is analogous.
Let $\gnnws_0$ be a graph neural network and let $F(\gnnws) = \gnnws \circ \gnnws_0$. We start by proving that $F$ is monotone.
If $\gnnws_1 \gnnorder \gnnws_2$, then $rel(\gnnws_1) \subseteq rel(\gnnws_2)$, and we can conclude that 
\[rel(\gnnws_0) \relcomp rel(\gnnws_1) \subseteq rel(\gnnws_0) \relcomp rel(\gnnws_2)\]
and therefore $F(\gnnws_1) \gnnorder F(\gnnws_2)$.

Next we prove the continuity of $F$. Let $Y$ be a non-empty chain, then
\begin{align*}
    rel(F(\bigsqcup Y)) &= rel((\bigsqcup Y) \circ \gnnws_0)\\
    &= rel(\gnnws_0) \relcomp rel(\bigsqcup Y)\\
    &= rel(\gnnws_0) \relcomp \bigcup \{rel(\gnnws) \mid \gnnws \in Y\} \\
    &= \bigcup\{rel(\gnnws_0) \relcomp rel(\gnnws) \mid \gnnws \in Y\}\\
    &= rel(\bigsqcup \{F(\gnnws) \mid \gnnws \in Y\})
\end{align*}
Thus $F(\bigsqcup Y) = \bigsqcup \{F(\gnnws) \mid \gnnws \in Y\}$.
\end{proof}

\begin{lemma}\label{lemma:cont-par}
    Parallel composition $\gnncomp$ is continuous in both its arguments.
\end{lemma}

\begin{proof}
    We prove that $\gnncomp$ is continuous in the first argument. The proof of continuity in the second argument is analogous. 
    Let $\gnnws_0$ be a graph neural network and let $F(\gnnws) = \gnnws \gnncomp \gnnws_0$. We start by proving that $F$ is monotone. 
    If $\gnnws_1 \gnnorder \gnnws_2$, then $rel(\gnnws_1) \subseteq rel(\gnnws_2)$, and we can conclude that 
    \[rel(\gnnws_0) \relpar rel(\gnnws_1) \subseteq rel(\gnnws_0) \relpar rel(\gnnws_2)\]
    and therefore $F(\gnnws_1) \gnnorder F(\gnnws_2)$.
    Next, we prove the continuity of $F$. Let $Y$ be a non-empty chain, then
    \begin{align*}
    rel(F(\bigsqcup Y)) &= rel((\bigsqcup Y) | \gnnws_0)\\
    &= rel(\bigsqcup Y) \frown rel(\gnnws_0) \\
    &= \bigcup \{rel(\gnnws) \mid \gnnws \in Y\} \frown rel(\gnnws_0) \\
    &= \bigcup\{rel(\gnnws) \frown rel(\gnnws_0) \mid \gnnws \in Y\}\\
    &= rel(\bigsqcup \{F(\gnnws) \mid \gnnws \in Y\})
\end{align*}

Thus $F(\bigsqcup Y) = \bigsqcup \{F(\gnnws) \mid \gnnws \in Y\}$.
    
\end{proof}

\begin{lemma}\label{lemma:cont-cond1}
     The function $cond$ is continuous in its first argument.
\end{lemma}

\begin{proof}
    Let $\gnnws_0, \gnnws_1$ be GNNs, and let $F(\gnnws) = cond(\gnnws, \gnnws_0, \gnnws_1)$. We start by proving that $F$ is monotone. We have to show that if $\gnnws_2 \gnnorder \gnnws_3$, then $F(\gnnws_2) \sqsubseteq F(\gnnws_3)$. We consider an arbitrary node-labeling function $\lblfunv$ and show that
    \[F(\gnnws_2)(\lblfunv) = \lblfunv' \implies F(\gnnws_3)(\lblfunv) = \lblfunv'\]
    If $\gnnws_2(\lblfunv) = \lambda v . \mathtt{True}$, then $F(\gnnws_2)(\lblfunv) = \gnnws_0(\lblfunv)$, and from $\gnnws_2 \gnnorder \gnnws_3$, we get that $F(\gnnws_3)(\lblfunv) = \gnnws_0(\lblfunv)$. Let's consider the case $\gnnws_2(\lblfunv) \neq \lambda v . \mathtt{True}$. Then $F(\gnnws_2)(\lblfunv) = \gnnws_1(\lblfunv)$ and the same is true for $F(\gnnws_3)(\lblfunv)$ so the result is immediate.

    To prove the continuity of $F$. Let $Y$ be a non-empty chain, and we will only show that $F(\bigsqcup Y) \gnnorder \bigsqcup \{F(\gnnws) \mid \gnnws \in Y \}$ since from the monotonicity of $F$ we just proved we can conclude $\bigsqcup \{F(\gnnws) \mid \gnnws \in Y \} \gnnorder  F(\bigsqcup Y)$.

    Then let's assume $F(\bigsqcup Y)(\lblfunv) = \lblfunv'$ and we have to determine a $\gnnws \in Y$ such that $F(\gnnws)(\lblfunv) = \lblfunv'$. If $(\bigsqcup Y)(\lblfunv) = \lambda v . \mathtt{True}$, we have $F(\bigsqcup Y)(\lblfunv) = \gnnws_1(\lblfunv)$ and it must be the case that $(\lblfunv, \lambda v . \mathtt{True}) \in rel(\bigsqcup Y)$. But since $rel(\bigsqcup Y) = \bigcup \{rel(\gnnws) \mid \gnnws \in Y\}$ there must exist a $\gnnws \in Y$ such that $\gnnws(\lblfunv) = \lambda v . \mathtt{True}$ and $F(\gnnws)(\lblfunv) = \gnnws_1(\lblfunv)$. The case where $(\bigsqcup Y)(\lblfunv) \neq \lambda v . \mathtt{True}$ is analogous.
\end{proof}

\begin{lemma}\label{lemma:cont-cond2}
    The function $cond$ is continuous in its second and third arguments.
\end{lemma}

\begin{proof}
    We prove that $cond$ is continuous in the second argument. The proof of continuity in the third argument is analogous.
    Let $\gnnws_t, \gnnws_0$ be GNNs, and let $F(\gnnws) = cond(\gnnws_t, \gnnws, \gnnws_0)$. We start by proving that $F$ is monotone. 


    If $\gnnws_t(\lblfunv) = \lambda v . \mathtt{True}$, then $F(\gnnws_1)(\lblfunv) = \gnnws_1(\lblfunv)$, and from $\gnnws_1 \gnnorder \gnnws_2$, we get that $\gnnws_1(\lblfunv) = \lblfunv' \implies \gnnws_2(\lblfunv) = \lblfunv'$. Therefore $F(\gnnws_2)(\lblfunv) = \gnnws_2(\lblfunv) = \lblfunv'$ proves our result. Let's consider the case $\gnnws_t(\lblfunv) \neq \lambda v . \mathtt{True}$. Then $F(\gnnws_1)(\lblfunv) = \gnnws_0(\lblfunv)$ and the same is true for $F(\gnnws_2)(\lblfunv)$ so the result is immediate.

    To prove the continuity of $F$. Let $Y$ be a non-empty chain, and we will only show that $F(\bigsqcup Y) \gnnorder \bigsqcup \{F(\gnnws) \mid \gnnws \in Y \}$ since from the monotonicity of $F$ we just proved we can conclude $\bigsqcup \{F(\gnnws) \mid \gnnws \in Y \} \gnnorder  F(\bigsqcup Y)$.

    Then let's assume $F(\bigsqcup Y)(\lblfunv) = \lblfunv'$ and we have to determine a $\gnnws \in Y$ such that $F(\gnnws)(\lblfunv) = \lblfunv'$. If $\gnnws_t(\lblfunv) \neq \lambda v . \mathtt{True}$, we have $F(\bigsqcup Y)(\lblfunv) = \gnnws_0(\lblfunv)$ and for all $\gnnws \in Y$ we have $F(\gnnws)(\lblfunv) = \gnnws_0(\lblfunv)$. If $\gnnws_t(\lblfunv) = \lambda v . \mathtt{True}$, we have $F(\bigsqcup Y)(\lblfunv) = (\bigsqcup Y)(\lblfunv) = \lblfunv'$ and therefore $(\lblfunv, \lblfunv') \in rel(\bigsqcup Y)$. But then, since $rel(\bigsqcup Y) = \bigcup \{rel(\gnnws) \mid \gnnws \in Y\}$, there must exist a $\gnnws \in Y$ such that $\gnnws(\lblfunv) = \lblfunv'$ and then $F(\gnnws)(\lblfunv) = \lblfunv'$.
\end{proof} 

\noindent Finally, we can show that the functional we pass to $FIX$ is indeed a continuous function.

\begin{theorem}\label{thm:F-cont}
    The functional $F(\gnnws) = cond(\lambda e . \lambda v . \ds{\mathcal{N}}(e)(v) \simop e(v), id, \gnnws \circ \ds{\mathcal{N}})$ is continuous.
\end{theorem}

\begin{proof}
    We can see $F$ as the composition of two functions 
    \[F = \lambda f . cond(\lambda e . \lambda v . \ds{\mathcal{N}}(e)(v) \simop e(v), f, id) \circ (\lambda f . f \circ \ds{\mathcal{N}})\]
    and since the function composition of continuous functions is itself continuous, we conclude its continuity from the continuity of function composition (Lemma~\ref{lemma:cont-circ}) and $cond$ (Lemma~\ref{lemma:cont-cond2}).
\end{proof}

We conclude this section with a proof of the well-definedness of $\ds{\cdot}$.

\begin{theorem}
Given a graph $\graph$, an edge-labeling function $\lblfune$, and under the assumption that the functions associated to any function symbols $\psi, \varphi, \sigma \in S$ are well-defined, the semantic interpretation function $\ds{\cdot}: \mathcal{N} \rightarrow \gnntypews$ is well-defined.
\end{theorem}

\begin{proof}
    The proof is by structural induction on $\mathcal{N}$.
    \begin{description}
        \item[The case $\id$:] The identity function is well-defined.
        \item[The case $\psi$:] The function that maps a node-labeling function $\lblfunv$ to $\lambda v . f_{\psi}(\lblfunv(\vSet), \lblfunv(v))$ is well-defined under the assumptions of the theorem.
        \item[The case $\lhd_{\sigma}^{\varphi}$:] Analogous to case $\psi$.
        \item[The case $\rhd_{\sigma}^{\varphi}$:] Analogous to case $\psi$.
        \item[The case $\mathcal{N}_1 ; \mathcal{N}_2$:] By induction hypothesis both $\ds{\mathcal{N}_1}$ and $\ds{\mathcal{N}_2}$ are well-defined, therefore their composition is also well-defined.
        \item[The case $\mathcal{N}_1 || \mathcal{N}_2$:] By induction hypothesis both $\ds{\mathcal{N}_1}$ and $\ds{\mathcal{N}_2}$ are well-defined, and so is their parallel composition.
        \item[The case $\mathcal{N}_1 \parcomp \mathcal{N}_2$:] By induction hypothesis both $\ds{\mathcal{N}_1}$ and $\ds{\mathcal{N}_2}$ are well-defined, and so are the projection functions $\pi_L, \pi_R$. Their composition then is clearly well-defined as well, and by the same argument for the case $\mathcal{N}_1 || \mathcal{N}_2$ their parallel composition is also well-defined.
        \item[The case $\mathcal{N}_1 \choiceop \mathcal{N}_2$:] By induction hypothesis both $\ds{\mathcal{N}_1}$ and $\ds{\mathcal{N}_2}$ are well-defined, and so are the projection functions $\pi_L, \pi_R$, and their composition with $\ds{\mathcal{N}_1}$ and $\ds{\mathcal{N}_2}$. The function $cond$ is therefore well-defined as well.
        \item[The case $\starop{\mathcal{N}}$:] By induction hypothesis $\ds{\mathcal{N}}$ is well-defined, and by Theorem~\ref{thm:F-cont} we have that $F(\phi) = cond(\lambda e . \lambda v . \ds{\mathcal{N}}(e)(v) \simop e(v), id, \gnnws \circ \ds{\mathcal{N}})$ is continuous. Then by Kleene's fixed-point theorem we have that $FIX(F)$ is well-defined.
    \end{description}
\end{proof}

\subsection{Operational semantics}
For the structural operational semantics, we consider terms of the form
\[ \soscomp{\mathcal{N}}{\lblfunv} \]
where $\mathcal{N}$ is a $\mG$ expression and $\lblfunv$ is a node-labeling function. The transition relation $\rightarrow_{\graph, \lblfune}$ is shown in Table~\ref{tab:structural_semantics} and we denote its reflexive-transitive closure as $\rightarrow_{\graph, \lblfune}^\star$. For any label type $T$ and any rational value $\epsilon \in \mathbb{Q}$, the function symbol $\psi_{\simop}$ in Table~\ref{tab:structural_semantics} is associated to a function $f_{\simop}: (T \times T)^\star \times (T \times T) \rightarrow \mathbb{B}$ defined as 
\[f_{\simop}(X, x) = x_1 \simop x_2\]
where $\simop$ is the predicate from Definition~\ref{def:denotational_semantics}.

\begin{table}[ht!]
\centering
\begin{tabular}{|lc|}
\hline & \\[1mm]

[\texttt{ID}] &
$\sosv{\id}{\lblfunv}{\lblfunv}$ \\[5mm]

[\texttt{APPLY}] &
$\sos{\psi}{\lblfunv}{\id}{\lambda v . f_{\psi}(\lblfunv(\vSet), \lblfunv(v))}$ \\[5mm]

[\texttt{PREIMG}] &
$\sos{\lhd_{\sigma}^{\varphi}}{\lblfunv}{\id}{\lambda v . f_\sigma([f_\varphi(\lblfunv(u), \lblfune((u, v)), \lblfunv(v)) \mid (u, v) \in \inset{\graph}{v}], \lblfunv(v))}$ \\[5mm]

[\texttt{POSTIMG}] &
$\sos{\rhd_{\sigma}^{\varphi}}{\lblfunv}{\id}{\lambda v . f_\sigma([f_\varphi(\lblfunv(u), \lblfune((v, u)), \lblfunv(v)) \mid (v, u) \in \outset{\graph}{v}], 
\lblfunv(v)) }$ \\[5mm]

[$\texttt{SEQ}_1$] & \inference{\sos{\mathcal{N}_1}{\lblfunv}{\mathcal{N}_1'}{\lblfunv'}}
{\sos{\mathcal{N}_1;\mathcal{N}_2}{\lblfunv}{\mathcal{N}_1';\mathcal{N}_2}{\lblfunv'}} \\[5mm] 

[$\texttt{SEQ}_2$] & $\sos{\id;\mathcal{N}_2}{\lblfunv}{\mathcal{N}_2}{\lblfunv}$ \\[5mm] 

[\texttt{SPLIT}] &
$\sos{\mathcal{N}_1 || \mathcal{N}_2}{\lblfunv}{\mathcal{N}_1 \parcomp \mathcal{N}_2}{\lblfunv | \lblfunv}$ \\[5mm]

[$\texttt{PAR}_1$] & \inference{\sos{\mathcal{N}_1}{\pi_L(\lblfunv)}{\mathcal{N}_1'}{\lblfunv'}}
{\sos{\mathcal{N}_1 \parcomp \mathcal{N}_2}{\lblfunv}{\mathcal{N}_1' \parcomp \mathcal{N}_2}{\lblfunv' | \pi_R(\lblfunv)}}\\[5mm] 

[$\texttt{PAR}_2$] & \inference{\sos{\mathcal{N}_2}{\pi_R(\lblfunv)}{\mathcal{N}_2'}{\lblfunv'}}
{\sos{\mathcal{N}_1 \parcomp \mathcal{N}_2}{\lblfunv}{\mathcal{N}_1 \parcomp \mathcal{N}_2'}{\pi_L(\lblfunv )| \lblfunv'}} \\[5mm] 

[\texttt{MERGE}] & $\sos{\id \parcomp \id}{\lblfunv}{\id}{\lblfunv}$ \\[5mm]

[$\mathtt{CHOICE}_1$] &
$\sos{ \mathcal{N}_1 \choiceop \mathcal{N}_2}{\lblfunv}{\mathcal{N}_1}{\pi_R(\lblfunv)}$ if $\pi_L(\lblfunv) = \lambda v . \mathtt{True}$ \\[5mm]

[$\mathtt{CHOICE}_2$] &
$\sos{\mathcal{N}_1 \choiceop \mathcal{N}_2}{\lblfunv}{\mathcal{N}_2}{\pi_R(\lblfunv)}$ if $\pi_L(\lblfunv) \neq \lambda v . \mathtt{True}$ \\[5mm]

[$\mathtt{STAR}$] &
$\sos{\starop{\mathcal{N}}}{\lblfunv}{((\id || \mathcal{N}) ; \psi_{\simop} || \id) ; (\id \choiceop \mathcal{N} ; \starop{\mathcal{N}})}{\lblfunv}$\\[5mm]

\hline
\end{tabular}
\caption{Structural operational semantics of $\mG$}
\label{tab:structural_semantics}
\end{table}

The rules $\texttt{PAR}_1$ and $\texttt{PAR}_2$ in Table~\ref{tab:structural_semantics} may give rise to multiple derivation sequences for a given $\mG$ expression. The next theorem shows that no matter the order of application of the structural semantics rules, the obtained graph neural network has the same input-output behavior, and we say that in this sense, the operational semantics are \textit{deterministic}.

\begin{theorem}[Strong Confluence]\label{th:confluence}
    If $\sos{\mathcal{N}}{\lblfunv}{\mathcal{N}_1}{\lblfunv_1}$ and $\sos{\mathcal{N}}{\lblfunv}{\mathcal{N}_2}{\lblfunv_2}$, then there exists a term $\mathcal{N}'$ and a labeling function $\lblfunv'$ such that $\sos[\star]{\mathcal{N}_1}{\lblfunv_1}{\mathcal{N}'}{\lblfunv'}$ and either $\sos{\mathcal{N}_2}{\lblfunv_2}{\mathcal{N}'}{\lblfunv'}$ or $\mathcal{N}' = \mathcal{N}_2$ and $\lblfunv' = \lblfunv_2$.
\end{theorem}

\begin{proof}
    By induction on the structure of the derivation sequences. The only interesting case is that for rules $\texttt{PAR}_1$ and $\texttt{PAR}_2$, which can both be applied on expressions of the form $\mathcal{N}_1 \parcomp \mathcal{N}_2$. Suppose that $\sos{\mathcal{N}_1 \parcomp \mathcal{N}_2}{\lblfunv}{\mathcal{N}_1' \parcomp \mathcal{N}_2}{\lblfunv_1 | \pi_R(\lblfunv)}$ by rule $\texttt{PAR}_1$ and that $\sos{\mathcal{N}_1 \parcomp \mathcal{N}_2}{\lblfunv}{\mathcal{N}_1 \parcomp \mathcal{N}_2'}{\pi_L(\lblfunv) | \lblfunv_2}$ by rule $\texttt{PAR}_2$.   
    Then the required term is $\mathcal{N}' = \mathcal{N}_1' \parcomp \mathcal{N}_2'$ and the labeling function is $\lblfunv' = \lblfunv_1 | \lblfunv_2$, because $\sos{\mathcal{N}_1' \parcomp \mathcal{N}_2}{\lblfunv_1 | \pi_R(\lblfunv)}{\mathcal{N}_1' \parcomp \mathcal{N}_2'}{\lblfunv_1 | \lblfunv_2}$ by rule $\texttt{PAR}_2$ and $\sos{\mathcal{N}_1 \parcomp \mathcal{N}_2'}{\pi_L(\lblfunv) | \lblfunv_2}{\mathcal{N}_1' \parcomp \mathcal{N}_2'}{\lblfunv_1 | \lblfunv_2}$ by rule $\texttt{PAR}_1$.
\end{proof}

Finally, we can define the semantic interpretation function on the structural operational semantics. 

\begin{definition}[Structural operational semantics]\label{def:sos}
Given a graph $\graph$ and an edge-labeling function $\lblfune$, we define the semantic interpretation function $\sosf{\cdot}: \mathcal{N} \rightarrow \gnntypews$ such that \[
    \sosf{\mathcal{N}}(\lblfunv) = \begin{cases}
        \lblfunv' & \text{if } \sosv[\star]{\mathcal{N}}{\lblfunv}{\lblfunv'} \\
        \texttt{undef} & \text{otherwise}
    \end{cases}
\]
\end{definition}

\subsection{Language macros}

We can enrich $\mG$ with a number of macros that simplifies the job of programming a graph neural network. In this section, we describe the means to define variables, functions, and shortcuts for the definition of if-then-else and while loop expressions. Some of these extensions require us to introduce an additional set of symbols $\mathcal{X} ::= X \mid Y \mid Z \mid \cdots$ to the language's syntax in order to denote variables and function names. We will also use function symbols $p_L$ and $p_R$ to denote left and right projections of node-labeling functions.

\paragraph{Variable assignments} It is useful sometimes to assign an entire expression to a single label, so that whenever that label occurs in a program, the referred expression is substituted to it. For this purpose, we introduce variable assignments in the form of \texttt{let} expressions. A \texttt{let} expression has the form $\texttt{let } \mathcal{X} = \mathcal{N} \texttt{ in } \mathcal{N}'$. The intuitive meaning of such expression is that all occurrences of the variable symbol $\mathcal{X}$ in the expression $\mathcal{N}'$ are substituted with the expression $\mathcal{N}$. This substitution is purely syntactical, and \texttt{let} expressions are simply rewritten as $\mathcal{N}_{[\mathcal{N}/\mathcal{X}]}'$.
Clearly, multiple \texttt{let} expressions can be chained, so that we can write $\texttt{let } \mathcal{X}_1 = \mathcal{N}_1, \mathcal{X}_2 = \mathcal{N}_2, \ldots, \mathcal{X}_k = \mathcal{N}_k \texttt{ in } \mathcal{N}$ as a shorthand for $\texttt{let } \mathcal{X}_1 = \mathcal{N}_1 \texttt{ in } (\texttt{let } \mathcal{X}_2 = \mathcal{N}_2 \texttt{ in } (\ldots (\texttt{let } \mathcal{X}_k = \mathcal{N}_k \texttt{ in } \mathcal{N})\ldots))$.
\paragraph{Function definitions} Similarly to variable assignments, we also consider the definition of functions. A \texttt{def} expression has the form $\texttt{def } \mathcal{X}(\mathcal{X}_1, \ldots \mathcal{X}_k) \{ \mathcal{N} \} \texttt{ in } \mathcal{N}'$. Its intuitive meaning is that whenever the variable symbol $\mathcal{X}$ followed by a parenthesized list of values $(\mathcal{N}_1, \ldots, \mathcal{N}_k)$ (which from now on we refer to as a \texttt{call} expression) occurs in $\mathcal{N}'$ the entire \texttt{call} expression is substituted with the expression $\mathcal{N}$ in which each variable symbol $\mathcal{X}_1, \ldots, \mathcal{X}_k$ has been substituted with the corresponding expression $\mathcal{N}_1, \ldots, \mathcal{N}_k$. As in the case of variables, this substitution is purely syntactical is rewritten as
%
\[\mathcal{N}_{[\mathcal{X}(\mathcal{N}_1, \ldots, \mathcal{N}_k)/\texttt{let } \mathcal{X}_1 = \mathcal{N}_1, \ldots, \mathcal{X}_k = \mathcal{N}_k \texttt{ in } \mathcal{N}]}'\]
\paragraph{If-then-else selection} A typical programming language construct is the \emph{if-then-else} selection operator that evaluates a Boolean condition and executes one out of two branches accordingly. In $\mG$ this operator can be implemented as a macro using the choice operator. We consider a graph neural network $\gnnws_1: \gnntype{T_1}{\mathbb{B}}$ denoted by an expression $\mathcal{N}_1$ to provide a Boolean node-labeling function, and then execute the graph neural network denoted by $\mathcal{N}_2$ if the Boolean labeling is \texttt{True} for every node, otherwise the GNN denoted by $\mathcal{N}_3$ is executed instead. Then we can introduce the term $\texttt{if } \mathcal{N}_1 \texttt{ then } \mathcal{N}_2 \texttt{ else } \mathcal{N}_3$ as a shorthand for $(\mathcal{N}_1 || \id) ; (\mathcal{N}_2 \choiceop \mathcal{N}_3)$.
%
%
\paragraph{Fixpoints with variables} Oftentimes we might want to include constant terms in our fixpoint computations, that is, include GNNs whose outputs do not depend on the current iterate solution, but only on the initial node labels as received by the star operator. To this end, we introduce a macro expression $\texttt{fix } \mathcal{X} = \mathcal{N}, \texttt{let } \mathcal{Y}_1 = \mathcal{N}_1, \ldots, \mathcal{Y}_k = \mathcal{N}_k \texttt{ in } \mathcal{N}'$ where the variable symbol $\mathcal{X}$ is the fixpoint variable and $\mathcal{N}$ is the GNN that computes the initial value for the fixpoint computation. The variables $\mathcal{Y}_i$ for $i \in 1, \ldots, k$ denote the values we want to pre-compute and stay constant during the computation. To clarify this behavior formally, we show how this macro can be expressed in the language.
We restrict the placement of variable symbols $\mathcal{X}, \mathcal{Y}$ to not appear on the right-hand side of a sequential composition expression. Then, this macro can be implemented as
\begin{align*}
&((\mathcal{N}_1 || \mathcal{N}_2 || \cdots || \mathcal{N}_k) || \mathcal{N});\\
&\starop{((p_1 || p_2 || \cdots || p_k) || \mathcal{N}_{[p_1/\mathcal{Y}_1, \ldots,  p_k/\mathcal{Y}_k, p_R/\mathcal{X}]}')};\\
&p_{R}
\end{align*}
where the terms $p_i$ denote a composition of left and right projections functions to obtain the node-labeling produced by $\ds{\mathcal{N}_i}$ in the parallel composition of $\mathcal{N}_i$ for $i \in 1, \ldots, k$. Indeed, it is even possible to statically infer the largest possible set of sub-expressions $\mathcal{N}_i$ that can be pre-computed, and therefore we can simply write $\texttt{fix } \mathcal{X} = \mathcal{N} \texttt{ in } \mathcal{N}'$ to have the same effect. This macro then essentially implements $\mu$-calculus style fixpoint computations, with the added flexibility of being able to set any initial value for the fixpoint computation.

In the case there is no such sub-expression $\mathcal{N}_i$, the macro reduces to
\[\mathcal{N};\starop{\mathcal{N}_{[\id/\mathcal{X}]}}\]
\paragraph{Fixed iteration loops} The basic fixpoint operator iterates until the output labeling is considered equal to the input labeling. At times the programmer might want to instead iterate for a fixed number of steps, that is, execute the body of the fixpoint expression only for $k \geq 1$ steps. For this purpose, we expand on the previous macro to introduce $\texttt{repeat } \mathcal{X} = \mathcal{N} \texttt{ in } \mathcal{N}' \texttt{ for } k$ with $k \in \mathbb{N}^{+}$ that is a shorthand for the $\mG$ expression
\begin{align*}
    &(\texttt{fix } \mathcal{X} = zero || \mathcal{N} \texttt{ in }\\
    &\qquad \texttt{if }  \mathcal{X} ; p_L ; <_k \texttt{ then }\\
    &   \qquad\qquad (\mathcal{X} ; p_L ; succ) || \mathcal{N}_{[\mathcal{X};p_R/\mathcal{X}]}' \\
    &\qquad\texttt{else }\\
    &   \qquad\qquad  \mathcal{X});p_R
\end{align*}
%
where $zero$ denotes a constant function that maps all node labels to 0, $<_k$ denotes a Boolean function that maps nodes to \texttt{True} if their (integer) label is smaller than $k$, and $succ$ denotes the successor function over natural numbers. Clearly, we can also have a \texttt{repeat} macro without having to specify a variable and an initial condition: $\texttt{repeat } \mathcal{N}' \texttt{ for } k$ can be rewritten as $\texttt{repeat } \mathcal{X} = \id \texttt{ in } \mathcal{N}' \texttt{ for } k$.

\section{Equivalence of operational and denotational semantics}\label{sec:equivalence}

Both the denotational semantics and operational semantics of $\mG$ define a mapping from node-labeling functions to node-labeling functions, i.e., a graph neural network. In the next theorem, we show that they define the \textit{same} graph neural network.

\begin{theorem}[Equivalence of denotational and structural operational semantics]

For any $\mG$ expression $\mathcal{N}$, any graph $\mathcal{G}$ and any edge-labeling function $\lblfune$, we have that 
\[ \ds{\mathcal{N}} = \sosf{\mathcal{N}} \]
\end{theorem}

\begin{proof}
Since both functions return GNNs, which belong to a partially ordered set $\gnntypews$ with ordering relation $\gnnorder$, it is sufficient to show that for any $\mG$ expression $\mathcal{N}$:
\begin{itemize}
    \item $\sosf{\mathcal{N}} \gnnorder \ds{\mathcal{N}}$
    \item $\ds{\mathcal{N}} \gnnorder \sosf{\mathcal{N}}$ 
\end{itemize}
Lemma~\ref{lemma:sos-to-ds} and Lemma~\ref{lemma:ds-to-sos} prove that both these conditions are satisfied by the semantics of $\mG$.
\end{proof}

\begin{lemma}\label{lemma:sos-to-ds}
For every expression $\mathcal{N}$ of $\mG$, we have $\sosf{\mathcal{N}} \gnnorder \ds{\mathcal{N}}$.
\end{lemma}

\begin{proof} 
We shall prove that for any expression $\mathcal{N}$ and any two node-labeling functions $\lblfunv, \lblfunv'$
\begin{equation}\label{eq:sos-to-ds}
\sosv[\star]{\mathcal{N}}{\lblfunv}{\lblfunv'} \implies \ds{\mathcal{N}}(\lblfunv) = \lblfunv'
\end{equation}

For that end, we will show that 
\begin{equation}\label{eq:value-step}
\sosv{\mathcal{N}}{\lblfunv}{\lblfunv'} \implies \ds{\mathcal{N}}(\lblfunv) = \lblfunv'
\end{equation}
\begin{equation}\label{eq:small-step}
\sos{\mathcal{N}}{\lblfunv}{\mathcal{N}'}{\lblfunv'}  \implies \ds{\mathcal{N}}(\lblfunv) = \ds{\mathcal{N}'}(\lblfunv')
\end{equation}

If Equation~\ref{eq:value-step} and Equation~\ref{eq:small-step} hold, the proof of Equation~\ref{eq:sos-to-ds} is a straightforward induction on the length $k$ of the derivation sequence $\sosv[k]{\mathcal{N}}{\lblfunv}{\lblfunv'}$. The proof of Equation~\ref{eq:value-step} and Equation~\ref{eq:small-step} is by induction on the shape of the derivation trees.

\begin{description}
\item[The case {\texttt{ID}}:] We have $\sosv{\id}{\lblfunv}{\lblfunv}$, and since $\ds{\id}(\lblfunv) = \lblfunv$, the implication holds.
\item[The case {\texttt{APPLY}}:] We have $\sos{\psi}{\lblfunv}{\id}{\lambda v . f_{\psi}(\lblfunv(\vSet), \lblfunv(v))}$, and since 
\begin{align*}
    \ds{\psi}(\lblfunv) &= \lambda v . f_{\psi}(\lblfunv(\vSet), \lblfunv(v))\\
                        &= \ds{\id}(\lambda v . f_{\psi}(\lblfunv(\vSet), \lblfunv(v)))
\end{align*} the implication holds.
\item[The case {\texttt{PREIMG}}:] Analogous to case \texttt{APPLY}.
\item[The case {\texttt{POSTIMG}}:] Analogous to case \texttt{APPLY}.
\item[The case {$\texttt{SEQ}_1$}:] Assume that $\sos{\mathcal{N}_1;\mathcal{N}_2}{\lblfunv}{\mathcal{N}_1';\mathcal{N}_2}{\lblfunv'}$ because $\sos{\mathcal{N}_1}{\lblfunv}{\mathcal{N}_1'}{\lblfunv'}$. By induction hypothesis we have that $\ds{\mathcal{N}_1}(\lblfunv) = \ds{\mathcal{N}_1'}(\lblfunv')$. Then we get 
\begin{align*}
    \ds{\mathcal{N}_1 ; \mathcal{N}_2}(\lblfunv)    &= \ds{\mathcal{N}_2}(\ds{\mathcal{N}_1}(\lblfunv)) \\
                                                    &= \ds{\mathcal{N}_2}(\ds{\mathcal{N}_1'}(\lblfunv')) \\
                                                    &= \ds{\mathcal{N}_1' ; \mathcal{N}_2}(\lblfunv')
\end{align*}

\item[The case {$\texttt{SEQ}_2$}:] We have that $\sos{\id;\mathcal{N}_2}{\lblfunv}{\mathcal{N}_2}{\lblfunv}$. Then we get 
\begin{align*}
    \ds{\id ; \mathcal{N}_2}(\lblfunv)    &= \ds{\mathcal{N}_2}(\ds{\id}(\lblfunv)) \\
                                          &= \ds{\mathcal{N}_2}(\lblfunv)
\end{align*}
\item[The case {\texttt{SPLIT}}:] We have $\sos{\mathcal{N}_1 || \mathcal{N}_2}{\lblfunv}{\mathcal{N}_1 \parcomp \mathcal{N}_2}{\lblfunv | \lblfunv}$. Then we get 
\begin{align*}
\ds{\mathcal{N}_1 || \mathcal{N}_2}(\lblfunv)   &= \ds{\mathcal{N}_1} \gnncomp \ds{\mathcal{N}_2} (\lblfunv)\\
                                                &= (\ds{\mathcal{N}_1} \circ \pi_L) \gnncomp (\ds{\mathcal{N}_2} \circ \pi_R) (\lblfunv | \lblfunv)\\
                                                &= \ds{\mathcal{N}_1 \parcomp \mathcal{N}_2}(\lblfunv | \lblfunv)
\end{align*}

\item[The case {$\texttt{PAR}_1$}:] We have $\sos{\mathcal{N}_1 \parcomp \mathcal{N}_2}{\lblfunv}{\mathcal{N}_1' \parcomp \mathcal{N}_2}{\lblfunv' | \pi_R(\lblfunv)}$ because $\sos{\mathcal{N}_1}{\pi_L(\lblfunv)}{\mathcal{N}_1'}{\lblfunv'}$. By induction hypothesis we have that $\ds{\mathcal{N}_1}(\pi_L(\lblfunv)) = \ds{\mathcal{N}_1'}(\lblfunv')$. Then we get 
\begin{align*}
    \ds{\mathcal{N}_1 \parcomp \mathcal{N}_2}(\lblfunv)    &= (\ds{\mathcal{N}_1} \circ \pi_L) \gnncomp (\ds{\mathcal{N}_2} \circ \pi_R) (\lblfunv) \\
                                                           &= \ds{\mathcal{N}_1}(\pi_L(\lblfunv)) | \ds{\mathcal{N}_2} (\pi_R (\lblfunv))\\
                                                           &= \ds{\mathcal{N}_1'}(\lblfunv') | \ds{\mathcal{N}_2} (\pi_R (\lblfunv))\\
                                                           &= (\ds{\mathcal{N}_1'} \circ \pi_L) \gnncomp (\ds{\mathcal{N}_2} \circ \pi_R) (\lblfunv' | \pi_R(\lblfunv))\\
                                                           &= \ds{\mathcal{N}_1' \parcomp \mathcal{N}_2}(\lblfunv' | \pi_R(\lblfunv))
\end{align*}

\item[The case {$\texttt{PAR}_2$}:] Analogous to case $\texttt{PAR}_1$.

\item[The case {\texttt{MERGE}}:] We have that $\sos{\id \parcomp \id}{\lblfunv}{\id}{\lblfunv}$. Then we get:
\begin{align*}
\ds{\id \parcomp \id}(\lblfunv) &= (\ds{\id} \circ \pi_L) \gnncomp (\ds{\id} \circ \pi_R) (\lblfunv)\\
                                &= \ds{\id}(\pi_L(\lblfunv)) | \ds{\id}(\pi_R(\lblfunv))\\
                                &= \pi_L(\lblfunv) | \pi_R(\lblfunv)\\
                                &= \lblfunv\\
                                &= \ds{\id}(\lblfunv)
\end{align*}

\item[The case {$\texttt{CHOICE}_1$}:] We have $\sos{\mathcal{N}_1 \choiceop \mathcal{N}_2}{\lblfunv}{\mathcal{N}_1}{\pi_R(\lblfunv)}$ because $\pi_L(\lblfunv) = \lambda v . \mathtt{True}$. Then we get
\begin{align*}
    \ds{\mathcal{N}_1 \choiceop \mathcal{N}_2}(\lblfunv) &= cond(\pi_L, \ds{\mathcal{N}_1} \circ \pi_R, \ds{\mathcal{N}_2} \circ \pi_R)(\lblfunv)\\
    &= \ds{\mathcal{N}_1}(\pi_R(\lblfunv))
\end{align*}

\item[The case {$\texttt{CHOICE}_2$}:] Analogous to case $\texttt{CHOICE}_1$.

\item[The case {\texttt{STAR}}:] We have \[\sos{\starop{\mathcal{N}}}{\lblfunv}{((\id || \mathcal{N}) ; \psi_{\simop} || \id) ; (\id \choiceop \mathcal{N} ; \starop{\mathcal{N}})}{\lblfunv}\] Let $F = \lambda \gnnws .  cond(\lambda e . \lambda v . \ds{\mathcal{N}}(e)(v) \simop e(v), id, \gnnws \circ \ds{\mathcal{N}})$, then we get

\begin{align*}
\ds{\starop{\mathcal{N}}}(\lblfunv) &= FIX(F)(\lblfunv) \\
&= F(FIX(F))(\lblfunv)\\
&= cond(\lambda e . \lambda v . \ds{\mathcal{N}}(e)(v) \simop e(v), id, FIX(F) \circ \ds{\mathcal{N}})(\lblfunv)\\
&= cond(\lambda e . \lambda v . \ds{\mathcal{N}}(e)(v) \simop e(v), id, \ds{\starop{\mathcal{N}}} \circ \ds{\mathcal{N}})(\lblfunv)\\
&= cond(\ds{(\mathcal{N} || \id) ; \psi_{\simop}}, \ds{\id}, \ds{\mathcal{N} ; \starop{\mathcal{N}}})\\
&= cond(\pi_L, \ds{\id} \circ \pi_R, \ds{\mathcal{N} ; \starop{\mathcal{N}}} \circ \pi_R) \circ \ds{(\mathcal{N} || \id); \psi_{\simop} || \id} (\lblfunv) \\
&= \ds{\id \choiceop \mathcal{N} ; \starop{\mathcal{N}}} \circ \ds{((\id || \mathcal{N}) ; \psi_{\simop} || \id)}(\lblfunv) \\
&= \ds{((\id || \mathcal{N}) ; \psi_{\simop} || \id) ; (\id \choiceop \mathcal{N} ; \starop{\mathcal{N}})}(\lblfunv)
\end{align*}

\end{description}
\end{proof}

\begin{lemma}\label{lemma:ds-to-sos}
For every expression $\mathcal{N}$ of $\mG$, we have $\ds{\mathcal{N}} \gnnorder \sosf{\mathcal{N}} $.
\end{lemma}

\begin{proof} 
We proceed by structural induction on the expression $\mathcal{N}$.

\begin{description}
\item[The case $\id$:] Immediate, as $\ds{\id}(\lblfunv) = \sosf{\id}(\lblfunv)$.
\item[The case $\psi$:] This case is simple as well, since $\ds{\psi}(\lblfunv) = \sosf{\psi}(\lblfunv)$ by applying rule \texttt{APPLY} followed by rule \texttt{ID}.
\item[The case $\lhd_{\sigma}^{\varphi}$:] Analogous to case $\psi$, but instead using the rules \texttt{PREIMG} and \texttt{ID}.
\item[The case $\rhd_{\sigma}^{\varphi}$:] Analogous to case $\psi$, but instead using the rules \texttt{POSTIMG} and \texttt{ID}.
\item[The case $\mathcal{N}_1 ; \mathcal{N}_2$:] By induction hypothesis we have $\ds{\mathcal{N}_1} \gnnorder \sosf{\mathcal{N}_1}$ and $\ds{\mathcal{N}_2} \gnnorder \sosf{\mathcal{N}_2}$. By the monotonicity of $\circ$ in both its arguments (Lemma~\ref{lemma:cont-circ}) and the induction hypothesis we get that
\begin{align*}
    \ds{\mathcal{N}_1 ; \mathcal{N}_2}  &= \ds{\mathcal{N}_2} \circ \ds{\mathcal{N}_1} \\
    &\gnnorder \sosf{\mathcal{N}_2} \circ \sosf{\mathcal{N}_1}
\end{align*}
To continue the proof, we need to prove that if $\sosv[k]{\mathcal{N}_1}{\lblfunv}{\lblfunv'}$ then $\sos[k]{\mathcal{N}_1 ; \mathcal{N}_2}{\lblfunv}{\mathcal{N}_2}{\lblfunv'}$. The proof is by induction on the derivation length $k$. For $k = 1$ we have $\sosv{\mathcal{N}_1}{\lblfunv}{\lblfunv'}$, which means that $\mathcal{N}_1 = \id$ and $\lblfunv = \lblfunv'$. Therefore, by applying rule $\texttt{SEQ}_2$ we get $\sos{\id ; \mathcal{N}_2}{\lblfunv}{\mathcal{N}_2}{\lblfunv'}$. For the inductive step, let's assume the lemma holds for $k$ and prove that it holds for $k + 1$. We assume that $\sosv[k+1]{\mathcal{N}_1}{\lblfunv}{\lblfunv'}$, which means that $\sos{\mathcal{N}_1}{\lblfunv}{\mathcal{N}_1'}{\lblfunv_1}$ for some $\mathcal{N}_1', \lblfunv_1$, and $\sosv[k]{\mathcal{N}_1'}{\lblfunv_1}{\lblfunv'}$. Therefore the first rule we can apply for the sequential composition of $\mathcal{N}_1$ and $\mathcal{N}_2$ is $\texttt{SEQ}_1$ with $\sos{\mathcal{N}_1}{\lblfunv}{\mathcal{N}_1'}{\lblfunv_1}$ as its premise, and we get the transition $\sos{\mathcal{N}_1 ; \mathcal{N}_2}{\lblfunv}{\mathcal{N}_1' ; \mathcal{N}_2}{\lblfunv_1}$. Now we apply the induction hypothesis $\sosv[k]{\mathcal{N}_1'}{\lblfunv_1}{\lblfunv'} \implies \sos[k]{\mathcal{N}_1' ; \mathcal{N}_2}{\lblfunv_1}{\mathcal{N}_2}{\lblfunv'}$ to obtain the desired result, having performed $k + 1$ steps in total.
Using this result, we can conclude that
\[\sosf{\mathcal{N}_2} \circ \sosf{\mathcal{N}_1} \gnnorder \sosf{\mathcal{N}_1 ; \mathcal{N}_2}\]
as required.

\item[The case $\mathcal{N}_1 || \mathcal{N}_2$:] By induction hypothesis we have $\ds{\mathcal{N}_1} \gnnorder \sosf{\mathcal{N}_1}$ and $\ds{\mathcal{N}_2} \gnnorder \sosf{\mathcal{N}_2}$. By the monotonicity of $\gnncomp$ in both its arguments (Lemma~\ref{lemma:cont-par}) and the induction hypothesis, we get that
\begin{align*}
    \ds{\mathcal{N}_1 || \mathcal{N}_2}  &= \ds{\mathcal{N}_1} \gnncomp \ds{\mathcal{N}_2} \\
                                         &\gnnorder \sosf{\mathcal{N}_1} \gnncomp \sosf{\mathcal{N}_2}
\end{align*}
%
%
%
To continue the proof, we show that if $\sosv[k_1]{\mathcal{N}_1}{\lblfunv_1}{\lblfunv_1'}$ and $\sosv[k_2]{\mathcal{N}_2}{\lblfunv_2}{\lblfunv_2'}$, then $\sosv[k_1 + k_2]{\mathcal{N}_1 \parcomp \mathcal{N}_2}{\lblfunv_1 | \lblfunv_2}{\lblfunv_1' | \lblfunv_2'}$. For $k_1 = k_2 = 1$, we have that $\mathcal{N}_1 = \mathcal{N}_2 = \id$, $\lblfunv_1 = \lblfunv_1'$ and $\lblfunv_2 = \lblfunv_2'$. Hence, we get $\sosv{\id \parcomp \id}{\lblfunv_1 | \lblfunv_2}{\sosv{\id}{\lblfunv_1 | \lblfunv_2}{\lblfunv_1 | \lblfunv_2}}$ by applying rule \texttt{MERGE} followed by rule \texttt{ID} and performing $1 + 1 = 2$ steps. For the inductive step, let's assume the implication holds for $k_1, k_2$ and prove it for $k_1 + 1, k_2 + 1$. We have $\sosv{\mathcal{N}_1}{\lblfunv_1}{\sosv[k_1]{\mathcal{N}_1'}{\lblfunv_1''}{\lblfunv_1'}}$ and $\sosv{\mathcal{N}_2}{\lblfunv_2}{\sosv[k_2]{\mathcal{N}_2'}{\lblfunv_2''}{\lblfunv_2'}}$ for some $\mathcal{N}_1', \mathcal{N}_2', \lblfunv_1'', \lblfunv_2''$. Then we can perform two steps $\sosv{\mathcal{N}_1 \parcomp \mathcal{N}_2}{\lblfunv_1 | \lblfunv_2}{\sos{\mathcal{N}_1' \parcomp \mathcal{N}_2}{\lblfunv_1'' | \lblfunv_2}{\mathcal{N}_1' \parcomp \mathcal{N}_2'}{\lblfunv_1'' | \lblfunv_2''}}$ by applying rules [$\texttt{PAR}_1$] and [$\texttt{PAR}_2$]. Then by applying the induction hypothesis we get our result $\sosv[k_1 + k_2]{\mathcal{N}_1' \parcomp \mathcal{N}_2'}{\lblfunv_1'' | \lblfunv_2''}{\lblfunv_1' | \lblfunv_2'}$ in $2 + k_1 + k_2 = (k_1 + 1) + (k_2 + 1)$ steps.
Using this result, we can conclude that
\begin{align*}
    \sosf{\mathcal{N}_2} \gnncomp \sosf{\mathcal{N}_1} &\gnnorder \sosf{\mathcal{N}_1 \parcomp \mathcal{N}_2} \circ (id \gnncomp id)\\
                                                       &= \sosf{\mathcal{N}_1 || \mathcal{N}_2}
\end{align*}
as required.
%
%
%
\item[The case $\mathcal{N}_1 \parcomp \mathcal{N}_2$:] By induction hypothesis we have $\ds{\mathcal{N}_1} \gnnorder \sosf{\mathcal{N}_1}$ and $\ds{\mathcal{N}_2} \gnnorder \sosf{\mathcal{N}_2}$. By the monotonicity of $\circ$ (Lemma~\ref{lemma:cont-circ}) and that of $\gnncomp$ in both its arguments (Lemma~\ref{lemma:cont-par}), and the induction hypothesis, we get that
\begin{align*}
    \ds{\mathcal{N}_1 \parcomp \mathcal{N}_2}  &= (\ds{\mathcal{N}_1} \circ \pi_L) \gnncomp (\ds{\mathcal{N}_1} \circ \pi_R) \\
                                         &\gnnorder (\sosf{\mathcal{N}_1} \circ \pi_L) \gnncomp (\sosf{\mathcal{N}_2} \circ \pi_R)\\
                                         &= \sosf{\mathcal{N}_1 \parcomp \mathcal{N}_2}
\end{align*}

\item[The case $\mathcal{N}_1 \choiceop \mathcal{N}_2$:] We have that $\ds{\mathcal{N}_1 \choiceop \mathcal{N}_2} = cond(\pi_L, \ds{\mathcal{N}_1} \circ \pi_R, \ds{\mathcal{N}_2} \circ \pi_R)$. By induction hypothesis we have that $\ds{\mathcal{N}_1} \gnnorder \sosf{\mathcal{N}_1}$ and $\ds{\mathcal{N}_2} \gnnorder \sosf{\mathcal{N}_2}$.
By the monotonicity of $\circ$ (Lemma~\ref{lemma:cont-circ}) and that of $cond$ in its second and third arguments (Lemma~\ref{lemma:cont-cond2}), and the induction hypothesis we have that
\begin{align*}
    \ds{\mathcal{N}_1 \choiceop \mathcal{N}_2}  &= cond(\pi_L, \ds{\mathcal{N}_1} \circ \pi_R, \ds{\mathcal{N}_2} \circ \pi_R) \\
    &\gnnorder cond(\pi_L, \sosf{\mathcal{N}_1} \circ \pi_R, \sosf{\mathcal{N}_2} \circ \pi_R)
\end{align*}
Now from the rules $\texttt{CHOICE}_1$ and $\texttt{CHOICE}_2$ it follows that
\begin{align*}
    \sosf{\mathcal{N}_1 \choiceop \mathcal{N}_2}(\lblfunv) = \sosf{\mathcal{N}_1}(\pi_R(\lblfunv)) & \text{ if } \pi_L(\lblfunv) = \lambda v . \mathtt{True} \\
    \sosf{\mathcal{N}_1 \choiceop \mathcal{N}_2}(\lblfunv) = \sosf{\mathcal{N}_2}(\pi_R(\lblfunv)) & \text{ if } \pi_L(\lblfunv) \neq \lambda v . \mathtt{True}
\end{align*}
therefore
\[cond(\pi_L, \sosf{\mathcal{N}_1} \circ \pi_R, \sosf{\mathcal{N}_2} \circ \pi_R) = \sosf{\mathcal{N}_1 \choiceop \mathcal{N}_2}\]


\item[The case $\starop{\mathcal{N}}$:] We have that $\ds{\starop{\mathcal{N}}} = FIX(F)$, where
\[F = \lambda \gnnws .  cond(\lambda e . \lambda v . \ds{\mathcal{N}}(e)(v) \simop e(v), id, \gnnws \circ \ds{\mathcal{N}})\] 
We start by proving that if $f: D \rightarrow D$ is a continuous function on a ccpo $(D, \sqsubseteq)$ and if $d \in D$ is such that $f(d) \sqsubseteq d$, then $FIX(f) \sqsubseteq d$. To prove this, we recall that by the monotonicity of $f$ it must be the case that $f(f(d)) \sqsubseteq f(d)$, and more generally, that $f^{n+1}(d) \sqsubseteq f^n(d), \forall n \geq 0$. Therefore $d$ is an upper bound of the chain $\{f^n(d) \mid n \geq 0\}$. Then, since $\bot \sqsubseteq d$, we have that $f^n(\bot) \sqsubseteq f^n(d), \forall n \geq 0$ and $d$ is also an upper bound for the chain $\{f^n(\bot) \mid n \geq 0\}$. Then $FIX(f) = \bigsqcup \{f^n(\bot) \mid n \geq 0\} \sqsubseteq d$ since the least fixed point of $f$ is the least upper bound of that chain.

This result proves that $FIX(F) \sqsubseteq \sosf{\starop{\mathcal{N}}}$, on the condition that we show that $F(\sosf{\starop{\mathcal{N}}}) \sqsubseteq \sosf{\starop{\mathcal{N}}}$.  We can prove this by first showing that
\begin{align*}
    \sosf{\starop{\mathcal{N}}} &= \sosf{((\mathcal{N} || \id) ; \psi_{\simop} || \id) ; (\id \choiceop \mathcal{N} ; \starop{\mathcal{N}})}\\
    &\sqsupseteq \sosf{(\id \choiceop \mathcal{N} ; \starop{\mathcal{N}})} \circ \sosf{((\mathcal{N} || \id) ; \psi_{\simop} || \id)}\\
    &= cond(\pi_L, \sosf{\id} \circ \pi_R,  \sosf{\mathcal{N} ; \starop{\mathcal{N}}} \circ \pi_R) \circ \sosf{((\mathcal{N} || \id) ; \psi_{\simop} || \id)}\\
    &= cond(\sosf{(\mathcal{N} || \id) ; \psi_{\simop}}, \sosf{\id}, \sosf{\mathcal{N} ; \starop{\mathcal{N}}})\\
    &\sqsupseteq  cond(\sosf{(\mathcal{N} || \id) ; \psi_{\simop}}, \sosf{\id}, \sosf{\starop{\mathcal{N}}} \circ \sosf{\mathcal{N}})
\end{align*}
The induction hypothesis gives us that $\ds{\mathcal{N}} \sqsubseteq \sosf{\mathcal{N}}$, and by the monotonicity of $\circ$ (Lemma~\ref{lemma:cont-circ}) and $cond$ (Lemmas~\ref{lemma:cont-cond1} and~\ref{lemma:cont-cond2}) we get
\begin{align*}
    \sosf{\starop{\mathcal{N}}} &\sqsupseteq  cond(\sosf{(\mathcal{N} || \id) ; \psi_{\simop}}, \sosf{\id}, \sosf{\starop{\mathcal{N}}} \circ \sosf{\mathcal{N}})\\
    &\sqsupseteq cond(\ds{(\mathcal{N} || \id) ; \psi_{\simop}}, id, \sosf{\starop{\mathcal{N}}} \circ \ds{\mathcal{N}})\\
    &= F(\sosf{\starop{\mathcal{N}}})
\end{align*}
%
\end{description}
\end{proof}

\section{Type soundness}\label{sec:types}
\begin{table}[ht]
\centering
\begin{tabular}{lc}
[\texttt{T-ID}] & $\id: \gnntype{T}{T}$ \\[5mm]

[\texttt{T-PSI}] & \inference{\Gamma |- f_{\psi}: T^{\star}_{1} \times T_1 \rightarrow T_2}{\psi: \gnntype{T_1}{T_2}} \\[5mm]

[\texttt{T-PRE}]  & \inference{\Gamma |- f_\varphi: T_1 \times T_e \times T_1 \rightarrow T_2 & \Gamma |- f_\sigma: T_{2}^{\star} \times T_1 \rightarrow T_3}{\lhd_{\sigma}^{\varphi}: \gnntype{T_1}{T_3}} \\[5mm]

[\texttt{T-POST}]  & \inference{\Gamma |- f_\varphi: T_1 \times T_e \times T_1 \rightarrow T_2 & \Gamma |- f_\sigma: T_{2}^{\star} \times T_1 \rightarrow T_3}{\rhd_{\sigma}^{\varphi}: \gnntype{T_1}{T_3}} \\[5mm]

[\texttt{T-SEQ}] & \inference{\mathcal{N}_1: \gnntype{T_1}{T_2} & \mathcal{N}_2: \gnntype{T_2}{T_3}}{\mathcal{N}_1;\mathcal{N}_2: \gnntype{T_1}{T_3}} \\[5mm]

[\texttt{T-PAR}] & \inference{\mathcal{N}_1: \gnntype{T}{T_1} & \mathcal{N}_2: \gnntype{T}{T_2}}{\mathcal{N}_1 || \mathcal{N}_2: \gnntype{T}{T_1 \times T_2}} \\[5mm]

[\texttt{T-PROD}] & \inference{\mathcal{N}_1: \gnntype{T_1}{T_1'} & \mathcal{N}_2: \gnntype{T_2}{T_2'}}{\mathcal{N}_1 \parcomp \mathcal{N}_2: \gnntype{T_1 \times T_2}{T_1' \times T_2'}} \\[5mm]

[\texttt{T-CHOICE}] & \inference{\mathcal{N}_1: \gnntype{T_1}{T_2} & \mathcal{N}_2: \gnntype{T_1}{T_2}}{\mathcal{N}_1 \choiceop \mathcal{N}_2: \gnntype{\mathbb{B} \times T_1}{T_2}}  \\[5mm]

[\texttt{T-STAR}] & \inference{\mathcal{N} : \gnntype{T}{T}}{\starop{\mathcal{N}}: \gnntype{T}{T}}  \\[5mm]

\end{tabular}
\caption{Typing rules of $\mG$.}
\label{Tab:gnn_types}
\end{table}
We say that a $\mG$ term is \emph{well-typed} whenever it can be typed with the rules of Table~\ref{Tab:gnn_types}. These rules guarantee that any \emph{well-typed} $\mG$ term defines a graph neural network. Next, we will be using these typing rules to prove the type soundness of the operational semantics of $\mG$. Proving that $\mG$ is type safe ensures that the GNNs we program with it handle the data types correctly and produce results of the expected type. Concretely, this means that well-typed GNNs can always progress to the next step of computation and that after each step the type of the output is preserved.
\begin{theorem}[Type soundness]
If $\mathcal{N}: \gnntype{T_1}{T_2}$, $\lblfunv : \lblfunvtype{T_1}$ and $\sosv[\star]{\mathcal{N}}{\lblfunv}{\lblfunv'}$, then $\lblfunv' : \lblfunvtype{T_2}$.
\end{theorem}
\begin{proof}
Follows from Lemma~\ref{lemma:progress} (Progress) and Lemma~\ref{lemma:preservation} (Preservation).
\end{proof}
\begin{lemma}[Progress]\label{lemma:progress}
Suppose that $\mathcal{N}$ is a well-typed term $\mathcal{N}: \gnntype{T_1}{T_2}$ and that $\lblfunv: \lblfunvtype{T_1}$, then either $\mathcal{N} = \id$ or there exists $\mathcal{N}', \lblfunv'$ such that $\sos{\mathcal{N}}{\lblfunv}{\mathcal{N}'}{\lblfunv'}$.
\end{lemma}
\begin{proof}
    By induction on typing derivations. It suffices to show that either $\mathcal{N} = \id$ or there exists an operational rule that can be applied.
    \begin{description}
        \item[Case \texttt{T-ID}:] Immediate, as $\mathcal{N} = \id$.
        \item[Case \texttt{T-PSI}:] We have $\mathcal{N} = \psi$ and therefore can use rule \texttt{APPLY}.
        \item[Case \texttt{T-PRE}:] We have $\mathcal{N} = \lhd_\sigma^\varphi$ and therefore can use rule \texttt{PREIMG}.
        \item[Case \texttt{T-POST}:] We have $\mathcal{N} = \rhd_\sigma^\varphi$ and therefore can use rule \texttt{POSTIMG}.
        \item[Case \texttt{T-SEQ}:] We have $\mathcal{N} = \mathcal{N}_1 ; \mathcal{N}_2$ where $\mathcal{N}_1 : \gnntype{T_1}{T}$ and $\mathcal{N}_2 : \gnntype{T}{T_2}$. By induction hypothesis $\mathcal{N}_1$ is either $\id$ or it is possible to apply a transition rule. In the former case, we can apply rule $\texttt{SEQ}_2$, while in the latter case we can apply rule $\texttt{SEQ}_1$.
        \item[Case \texttt{T-PAR}:]  We have $\mathcal{N} = \mathcal{N}_1 || \mathcal{N}_2$ where $\mathcal{N}_1 : \gnntype{T}{T_1}$ and $\mathcal{N}_2 : \gnntype{T}{T_2}$. In this case, we can use rule \texttt{SPLIT}.
        \item[Case \texttt{T-PROD}:] We have $\mathcal{N} = \mathcal{N}_1 \parcomp \mathcal{N}_2$ where $\mathcal{N}_1 : \gnntype{T_1}{T_1'}$ and $\mathcal{N}_2 : \gnntype{T_2}{T_2'}$. By induction hypothesis $\mathcal{N}_1$ and $\mathcal{N}_2$ are either $\id$ or is it possible to apply a transition rule. Then,
        \begin{itemize}
            \item if $\mathcal{N}_1 = \id$ and $\mathcal{N}_2 = \id$, we can apply rule \texttt{MERGE}.
            \item if $\mathcal{N}_1 \neq \id$ and $\mathcal{N}_2 = \id$, we can apply rule $\texttt{PAR}_1$.
            \item if $\mathcal{N}_1 = \id$ and   $\mathcal{N}_2 \neq \id$, we can apply rule $\texttt{PAR}_2$.
            \item if $\mathcal{N}_1 \neq \id$ and $\mathcal{N}_2 \neq \id$, we can apply either rule $\texttt{PAR}_1$ or rule $\texttt{PAR}_2$.
        \end{itemize}
        \item[Case \texttt{T-CHOICE}:] We have $\mathcal{N} = \mathcal{N}_1 \choiceop \mathcal{N}_2 : \gnntype{\mathbb{B} \times T_1}{T_2}$ where $\mathcal{N}_1 : \gnntype{T_1}{T_2}$ and $\mathcal{N}_2 : \gnntype{T_1}{T_2}$. Furthermore, we have that $\lblfunv: \lblfunvtype{\mathbb{B} \times T_1}$. If $\pi_L(\lblfunv)$ is the constant node-labeling function that maps every node to \texttt{True}, we apply rule $\texttt{CHOICE}_1$. In any other case we apply rule $\texttt{CHOICE}_2$.
        \item[Case \texttt{T-STAR}:]  We have $\mathcal{N} = \starop{\mathcal{N}_1}$ where $\mathcal{N}_1 : \gnntype{T}{T}$. We can apply rule \texttt{STAR}.
    \end{description}
\end{proof}
\begin{lemma}[Preservation]\label{lemma:preservation}
If $\mathcal{N}: \gnntype{T_1}{T_2}$, $\lblfunv: \lblfunvtype{T_1}$, and $\sosv{\mathcal{N}}{\lblfunv}{\gamma}$
\begin{itemize}
    \item if $\gamma = \soscomp{\mathcal{N}'}{\lblfunv'}$, then $\mathcal{N}': \gnntype{T'}{T_2}$ and $\lblfunv': \lblfunvtype{T'}$, for some label type $T'$.
    \item if $\gamma = \lblfunv'$, then $\lblfunv': \lblfunvtype{T_2}$.
\end{itemize}

\end{lemma}


\begin{proof}
    By induction on typing derivations.
        \begin{description}
        \item[Case \texttt{T-ID}:] We have $\mathcal{N} = \id : \gnntype{T}{T}$ and $\lblfunv: \lblfunvtype{T}$. The only applicable rule is \texttt{ID} by which $\sosv{\id}{\lblfunv}{\lblfunv} : \lblfunvtype{T}$.
        
        \item[Case \texttt{T-PSI}:] We have $\psi: \gnntype{T_1}{T_2}$ and $\lblfunv: \lblfunvtype{T_1}$, with $f_\psi: T_1^\star \times T_1 \rightarrow T_2$. The only rule that can be applied is \texttt{APPLY} which gives $\mathcal{N}' = \id$ and $\lblfunv' = \lambda v . f_{\psi}(\lblfunv(\vSet), \lblfunv(v))$. Then $\lblfunv'$ has type $\lblfunvtype{T_2}$ and $\id$ has type $\gnntype{T_2}{T_2}$ by rule \texttt{T-ID}.
        
        \item[Case \texttt{T-PRE}:] We have $\lhd_\sigma^\varphi : \gnntype{T_1}{T_3}$ and $\lblfunv: \lblfunvtype{T_1}$, with $f_\varphi: T_1 \times T_e \times T_1 \rightarrow T_2$ and $f_\sigma: T_2^\star \times T_1 \rightarrow T_3$. The only rule that can be applied is \texttt{PREIMG} which gives $\mathcal{N}' = \id$ and $\lblfunv' = \lambda v . f_\sigma([f_\varphi(\lblfunv(u), \lblfune((u, v)), \lblfunv(v)) \mid (u, v) \in \inset{\graph}{v}], \lblfunv(v))$. Then $\lblfunv'$ has type $\lblfunvtype{T_3}$ and $\id$ has type $\gnntype{T_3}{T_3}$ by rule \texttt{T-ID}.
        
        \item[Case \texttt{T-POST}:] Analogous to \texttt{T-PRE}, but using rule \texttt{POSTIMG} instead.
        
        \item[Case \texttt{T-SEQ}:] We have $\mathcal{N}_1 ; \mathcal{N}_2 : \gnntype{T_1}{T_3}$ and $\lblfunv : \lblfunvtype{T_1}$. We also have sub-derivations with conclusions that $\mathcal{N}_1 : \gnntype{T_1}{T_2}$ and $\mathcal{N}_2 : \gnntype{T_2}{T_3}$. There are two possible rules that can be applied, $\texttt{SEQ}_1$ and $\texttt{SEQ}_2$. 
        \begin{itemize}
            \item In the case of $\texttt{SEQ}_1$, we have $\mathcal{N}' = \mathcal{N}_1' ; \mathcal{N}_2$. By applying the induction hypothesis on $\mathcal{N}_1$ we get that $\lblfunv' : \lblfunvtype{T}$ and $\mathcal{N}_1' : \gnntype{T}{T_2}$ for some node-labeling type $\lblfunvtype{T}$. Then we can apply rule \texttt{T-SEQ} on $\mathcal{N}_1' ; \mathcal{N}_2$ to get its type $\gnntype{T}{T_3}$.
            \item In the case of $\texttt{SEQ}_2$, then $\mathcal{N}_1 = \id$, $T_1 = T_2$ and we have $\mathcal{N}' = \mathcal{N}_2$, which we know to have type $\gnntype{T_1}{T_3}$, and $\lblfunv' = \lblfunv$, which we know to have type $\lblfunvtype{T_1}$. 
        \end{itemize}
        
        \item[Case \texttt{T-PAR}:] We have $\mathcal{N}_1 || \mathcal{N}_2 : \gnntype{T}{T_1 \times T_2}$ and $\lblfunv : \lblfunvtype{T}$. We also have sub-derivations with conclusions that $\mathcal{N}_1 : \gnntype{T}{T_1}$ and $\mathcal{N}_2 : \gnntype{T}{T_2}$. The only rule that can be applied is \texttt{SPLIT}. Then $\lblfunv' = \lblfunv | \lblfunv : \lblfunvtype{T \times T}$ and $\mathcal{N}' = \mathcal{N}_1 \parcomp \mathcal{N}_2$ which by rule \texttt{T-PROD} has type $\gnntype{T \times T}{T_1 \times T_2}$.

        \item[Case \texttt{T-PROD}:] We have $\mathcal{N}_1 \parcomp \mathcal{N}_2 : \gnntype{T_1 \times T_2}{T_1' \times T_2'}$ and $\lblfunv : \lblfunvtype{T_1 \times T_2}$. We also have sub-derivations with conclusions that $\mathcal{N}_1 : \gnntype{T_1}{T_1'}$ and $\mathcal{N}_2 : \gnntype{T_2}{T_2'}$. There are three possible rules that can be applied, $\texttt{PAR}_1, \texttt{PAR}_2$, and $\texttt{MERGE}$.
        \begin{itemize}
            \item In the case of $\texttt{PAR}_1$, we have $\mathcal{N}' = \mathcal{N}_1' \parcomp \mathcal{N}_2$ and $\lblfunv' = \lblfunv_1 | \pi_R(\lblfunv)$. By applying the induction hypothesis on $\mathcal{N}_1$ we get that $\lblfunv_1 : \lblfunvtype{T}$ and $\mathcal{N}_1' : \gnntype{T}{T_1'}$ for some node-labeling type $\lblfunvtype{T}$. Then we can apply rule \texttt{T-PROD} on $\mathcal{N}_1' \parcomp \mathcal{N}_2$ to get its type $\gnntype{T \times T_2}{T_1' \times T_2'}$, and $\lblfunv': \lblfunvtype{T \times T_2}$.
            \item The case of $\texttt{PAR}_2$ is analogous to that of $\texttt{PAR}_1$.
            \item In the case of $\texttt{MERGE}$, we have $\mathcal{N}' = \id$ and $\lblfunv' = \lblfunv$. Furthermore, this rule could have been applied only if $\mathcal{N}_1 = \mathcal{N}_2 = \id$, which means that $T_1 = T_1'$ and $T_2 = T_2'$. Then by rule \texttt{T-ID} we get that $\id: \gnntype{T_1' \times T_2'}{T_1' \times T_2'}$.
        \end{itemize}

        \item[Case \texttt{T-CHOICE}:] We have $\mathcal{N}_1 \choiceop \mathcal{N}_2: \gnntype{\mathbb{B} \times T_1}{T_2}$ and $\lblfunv: \lblfunvtype{\mathbb{B} \times T_1}$.  We also have sub-derivations with conclusions that $\mathcal{N}_1 : \gnntype{T_1}{T_2}$ and $\mathcal{N}_2 : \gnntype{T_1}{T_2}$. There are two possible rules that can be applied, $\texttt{CHOICE}_1$ and $\texttt{CHOICE}_2$.
        \begin{itemize}
            \item In case $\texttt{CHOICE}_1$, we have $\lblfunv' = \pi_R(\lblfunv) : \lblfunvtype{T_1}$ and $\mathcal{N}' = \mathcal{N}_1$ which we know to have type $\gnntype{T_1}{T_2}$.
            \item In case $\texttt{CHOICE}_2$, we have $\lblfunv' = \pi_R(\lblfunv) : \lblfunvtype{T_1}$ and $\mathcal{N}' = \mathcal{N}_2$ which we know to have type $\gnntype{T_1}{T_2}$.
        \end{itemize}

        \item[Case \texttt{T-STAR}:]  We have $\starop{\mathcal{N}}: \gnntype{T}{T}$ and $\lblfunv: \lblfunvtype{T}$. We also have a subderivation with the conclusion that $\mathcal{N}: \gnntype{T}{T}$. The only rule that can be applied is \texttt{STAR}. Then $\lblfunv' = \lblfunv : \lblfunvtype{T}$ and $\mathcal{N}' = ((\id || \mathcal{N}) ; \psi_{\simop} || \id) ; (\id \choiceop \mathcal{N} ; \starop{\mathcal{N}})$.  To type this term, first we infer the type of the left-hand side of the sequential composition
        \[\resizebox{0.94\textwidth}{!}{
        \inference[\texttt{T-PAR}]{
        \inference[\texttt{T-SEQ}]{
        \inference[\texttt{T-PAR}]{
        \inference[\texttt{T-ID}]{}{\id: \gnntype{T}{T}} & \mathcal{N}: \gnntype{T}{T}
        }{\id || \mathcal{N} : \gnntype{T}{T \times T}}
        &
        \inference[\texttt{T-PSI}]{\Gamma |- f_{\psi_{\simop}}: (T \times T)^{*} \times (T \times T) \rightarrow \mathbb{B}}{\psi_{\simop}: \gnntype{T \times T}{\mathbb{B}}}
        }{(\id || \mathcal{N}) ; \psi_{\simop} : \gnntype{T}{\mathbb{B}}}
        &
        \inference[\texttt{T-ID}]{}{\id: \gnntype{T}{T}}
        }{(\id || \mathcal{N}) ; \psi_{\simop} || \id : \gnntype{T}{\mathbb{B} \times T}}
        }\]
        We proceed likewise for the right-hand side
        \[
        \inference[\texttt{T-CHOICE}]{
        \inference[\texttt{T-ID}]{}{\id: \gnntype{T}{T}}
        &
        \inference[\texttt{T-SEQ}]{\mathcal{N}: \gnntype{T}{T}
        &
        \starop{\mathcal{N}} : \gnntype{T}{T}}{\mathcal{N} ; \starop{\mathcal{N}} : \gnntype{T}{T}}
        }{\id \choiceop \mathcal{N} ; \starop{\mathcal{N}} : \gnntype{\mathbb{B} \times T}{T}}\]
        Finally, we can apply rule \texttt{T-SEQ} to infer
        \[
        \inference[\texttt{T-SEQ}]{(\id || \mathcal{N}) ; \psi_{\simop} || \id : \gnntype{T}{\mathbb{B} \times T} & \id \choiceop \mathcal{N} ; \starop{\mathcal{N}} : \gnntype{\mathbb{B} \times T}{T}
        }
        {((\id || \mathcal{N}) ; \psi_{\simop} || \id) ; (\id \choiceop \mathcal{N} ; \starop{\mathcal{N}}): \gnntype{T}{T}}\]
        as required.
    \end{description}
\end{proof}

\section{A graphical representation}\label{sec:graphics}
In this section, we present a graphical representation for $\mG$ programs and expressions. Using such graphical notation makes it easier to understand and communicate the structure and purpose of a $\mG$ program, as the programmer might easily lose track of label types and sizes when developing in textual form. The basic terms of $\mG$ can be represented graphically as boxes, as shown in Figure~\ref{fig:basic-repr}, while the composite terms of the language are shown in Figure~\ref{fig:comp-repr}. Each term has exactly one input node-labeling function and one output node-labeling function, and the graphical representation is built top-down starting from the main (usually composite) term and substituting each box with either a basic term or a composite one; in this last case, this procedure is repeated recursively until all boxes in the figure are basic terms. As an illustrative example, let $\psi_1, \psi_2, \psi_3, \varphi, \sigma$ be function symbols. Figure~\ref{fig:example} shows how we can represent graphically the $\mu\mathcal{G}$ expression $(\psi_1 || \psi_2);\psi_3;\lhd_{\sigma}^{\varphi}$. 

\begin{figure}[ht]
    \centering
    \includegraphics[scale=0.33]{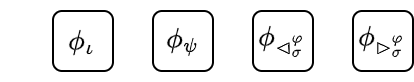}
    \caption{Graphical representation of basic $\mG$ expressions.}
    \label{fig:basic-repr}
    \Description{Three square boxes labelled as $\phi_{\psi}$, $\phi_{\lhd_{\sigma}^{\varphi}}$, and $\phi_{\rhd_{\sigma}^{\varphi}}$}
\end{figure}

\begin{figure}[ht]
\centering
    \subfloat[The sequential composition of graph neural networks $\phi_1$ and $\phi_2$.]{
        \includegraphics[scale=0.25]{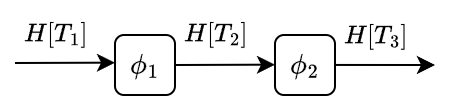}
    } \qquad
    \subfloat[The parallel composition of graph neural networks $\phi_1$ and $\phi_2$.]{
        \includegraphics[scale=0.25]{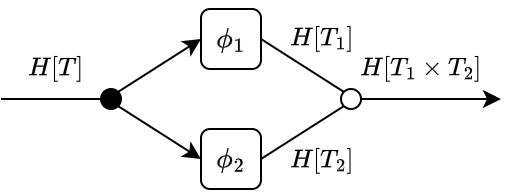}
    }

    \subfloat[The choice operator for the graph neural networks $\phi_1$ and $\phi_2$.]{
        \includegraphics[scale=0.25]{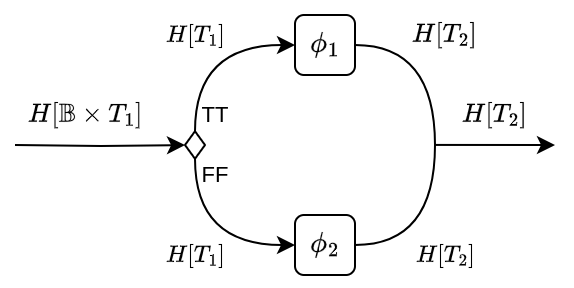}
    }

    \subfloat[The star graph neural network that computes the fixed point of $\phi_1$.]{
    \includegraphics[scale=0.25]{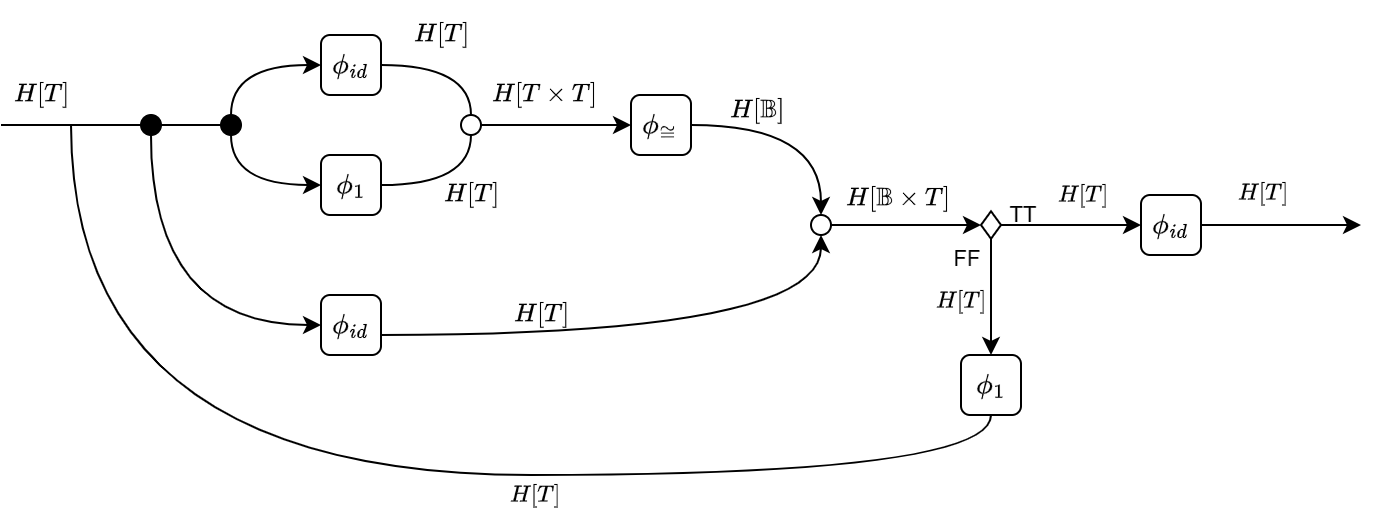}
    }
    
    \caption{Graphical representation of the possible compositions of $\mG$ expressions.}
    \label{fig:comp-repr}
    \Description{Four subfigures showing (1) the sequential composition of GNNs, (2) the parallel composition of GNNs, (3) the choice operator for GNNs, and (4) the star operator for GNNs.}
\end{figure}

\begin{figure}[ht]
    \centering
    \includegraphics[scale=0.25]{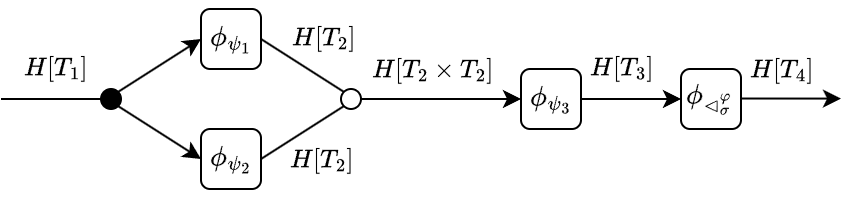}
    \caption{Example graphical representation of the $\mu\mathcal{G}$ expression $(\psi_1 || \psi_2);\psi_3;\lhd_{\sigma}^{\varphi}$.}
    \label{fig:example}
    \Description{The incoming arrow from the left splits into two and enters boxes $\phi_{\psi_{1}}$ and $\phi_{\psi_{2}}$. Then the arrows coming out of these boxes merge again and enter a third box $\phi_{\psi_3}$. Finally, the arrow enters a box $\phi_{\lhd_{\sigma}^{\varphi}}$ and exits pointing to the right.}
\end{figure}

In the graphical representation for $\mG$ expressions we have used two auxiliary functions and a choice operator which we describe in turn:
\begin{itemize}
\item The function $\bullet: \lblfunvtype{T_1} \rightarrow \lblfunvtype{T_1} \times \lblfunvtype{T_1}, \bullet(\lblfunv) = (\lblfunv, \lblfunv)$ maps a node-labeling function to a tuple containing two copies of itself.
\item The function $\circ: \lblfunvtype{T_1} \times \lblfunvtype{T_2} \rightarrow \lblfunvtype{T_1 \times T_2}$, $\circ(\lblfunv_1, \lblfunv_2) = \lblfunv_1 | \lblfunv_2$ maps two node-labeling functions to their parallel composition.
\end{itemize}
The diamond symbol $\diamond$ represents the choice operator that selects the \texttt{True} branch if the Boolean node-labeling function it receives in the left projection of its input labels each node in the graph with a \texttt{True} value, otherwise it selects the \texttt{False} branch. The selected branch receives the right projection of the node-labeling function sent to the choice operator.

\section{Implementation}\label{sec:implementation}
The $\mG$ language is available as a Python library called \textsc{libmg}, which we described in detail in a previous work~\cite{belenchia_libmg_2024}. The library was developed on top of Spektral~\cite{grattarola_graph_2021}, a graph neural network library based on TensorFlow~\cite{tensorflow_2015}. The main advantage of using a framework such as TensorFlow is that $\mG$ programs are compiled to graph neural networks that can be trained and executed seamlessly on CPUs, GPUs, or any other hardware supported by TensorFlow without changing the source code. TensorFlow also supports these operations in a distributed setting.

The functionalities of \textsc{libmg} include the typical formal language support tools in the form of a parser, unparser and normalizer for $\mG$ expressions. The main feature that is provided consists in the $\mG$ compiler that transforms a $\mG$ program into a TensorFlow \texttt{Model} instance. The compiler can be set up to automatically tabulate sub-expressions that occur multiple times in the same program. The library also implements basic graph visualization functions, that allows to view in any web browser the input or output graphs, or intermediate representation produced by a $\mG$ model. Finally, a simple explanatory algorithm is provided that computes the sub-graph that was used by a $\mG$ model to compute the label for a given node.

Using \textsc{libmg}, the $\psi, \varphi$, and $\sigma$ functions are defined by instantiating the corresponding classes and are stored in dictionaries. These dictionaries, which define the mapping between function names and the functions themselves, are used to create a $\mG$ compiler instance. The $\mG$ compiler can then parse $\mG$ programs containing these function names and return the GNN that specified by the program.


\section{Evaluation}\label{sec:evaluation}
We evaluate $\mG$ both in terms of its expressiveness, meaning that it can be used to define most, if not all, of the graph neural network architectures that have been developed over the years, and also as a specification language for the development of graph neural networks for specific tasks. As an example of the latter use case, in Section~\ref{sec:mc} we describe an updated version of the CTL model checker developed using $\mG$ introduced in our previous work~\cite{belenchia_implementing_2023, belenchia_libmg_2024}. Then, in Section~\ref{sec:gnn-examples}, we show how to define some of the most well-known GNN architectures using our language.

\subsection{Model Checking}\label{sec:mc}
We evaluated $\mG$ in a previous work by using it to define a GNN that performs explicit-state global CTL model checking on Kripke structures. For that use case, the set of functions that we used is shown in Table~\ref{tab:ctl-functions}, while in Table~\ref{tab:ctl-mg} we show the translation function from CTL to the $\mG$ expression denoting the GNN that verifies it.

\begin{table}[tbp]
    \centering
    \begin{tabular}{ll}
    $\psi_{p}: (2^{AP})^{\star} \times 2^{AP} \rightarrow \mathbb{B}$ &  $\psi_p(X, x) = \begin{cases}\mathtt{True} & \text{if } p \in x\\ \mathtt{False} & \text{otherwise} \end{cases} (\forall p \in AP) $\\ 
    $\psi_{tt}: (2^{AP})^{\star} \times 2^{AP} \rightarrow \mathbb{B}$ &  $\psi_{tt}(X, x) = \mathtt{True} $\\ 
    $\psi_{ff}: (2^{AP})^{\star} \times 2^{AP} \rightarrow \mathbb{B}$ &  $\psi_{ff}(X, x) = \mathtt{False} $\\ 
    $\psi_{\neg}: (\mathbb{B})^{\star} \times \mathbb{B} \rightarrow \mathbb{B}$     & $\psi_{\neg}(X, x) = \neg x$ \\ 
    $\psi_{\land}: (\mathbb{B}^2)^{\star} \times \mathbb{B}^2 \rightarrow \mathbb{B}$     & $\psi_{\land}(X, x) = x_1 \land x_2$ \\ 
    $\psi_{\lor}: (\mathbb{B}^2)^{\star} \times \mathbb{B}^2 \rightarrow \mathbb{B}$    & $\psi_{\lor}(X, x) = x_1 \lor x_2$ \\ 
    $\varphi_{id}: \mathbb{B} \times () \times \mathbb{B} \rightarrow \mathbb{B}$    & $\varphi_{id}(i, e, j) = i$ \\ 
    $\sigma_{\lor}: \mathbb{B}^{\star} \times \mathbb{B} \rightarrow \mathbb{B}$    & $\sigma_{\lor}(M, x) = \bigvee_i M_i$ \\ 
    \end{tabular}
    \caption{The $\psi, \varphi$, and $\sigma$ functions used for CTL model checking. The subscripts indicate the function symbols that will be used in the $\mG$ expressions. For each $p \in AP$, we consider a corresponding function $\psi_p$.}
    \label{tab:ctl-functions}
\end{table}

\begin{table}[tbp]
    \begin{align*}
    \mathcal{M}(p) &= p \\
    \mathcal{M}(\neg \phi) &= \mathcal{M}(\phi);\neg \\
    \mathcal{M}(\phi_1 \land \phi_2) &= (\mathcal{M}(\phi_1) || \mathcal{M}(\phi_2));\land \\
    \mathcal{M}(\texttt{EX}\; \phi) &= \mathcal{M}(\phi);\rhd^{id}_{\lor}\\
    \mathcal{M}(\texttt{EG}\; \phi) &= \texttt{fix } X = tt \texttt{ in } (\mathcal{M}(\phi) || X ; \rhd^{id}_{\lor});\land\\
    \mathcal{M}(\texttt{E}\; \phi_1 \;\texttt{U}\; \phi_2) &= \texttt{fix } X = ff \texttt{ in } (\mathcal{M}(\phi_2) || (\mathcal{M}(\phi_1) || X ; \rhd^{id}_{\lor});\lor);\land \\
    \end{align*}
    \caption{Recursive definition of the translation function $\mathcal{M}: \phi \rightarrow \mathcal{N}$.}
    \label{tab:ctl-mg}
\end{table}

The $\mG$ model checker implementation was compared to two other explicit-state model checkers, pyModelChecking\footnote{\url{https://github.com/albertocasagrande/pyModelChecking}} and the widely used mCRL2 toolset~\cite{mucrl2}. We used some of the benchmarks from the Model Checking Contest 2022~\cite{abcdg2019}, namely those for which we could explicitly generate the reachability graph on our machine, shown in Table~\ref{tab:models}. The $\mG$ model checker was set up to either verify all the formulas in parallel (the \textit{full} setup), or one by one as would have been the case for the other model checkers (the \textit{split} setup). In the former case, there could have been additional advantages due to the tabulation of common sub-formulas across multiple formulas. The experimental results are shown in Figure~\ref{fig:evaluation}, and we can notice that for the larger Kripke structures (>17000 states) there is an evident speed advantage of $\mG$ compared to the other model checkers (in some instances, by an order of magnitude). On the twelfth Petri Net, the full setup runs out of memory because the size of the formulas also has an influence on the amount of memory that is used by the GNN, and in that case their parallel composition exceeded the GPU memory.

\begin{table}[tbp]
\centering
\begin{tabular}{
  | c
  | l 
  | c 
  | c |  }
\hline
\textbf{ID} & \textbf{Petri Net} & \textbf{\#nodes in RG} & \textbf{\#edges in RG} \\ \hline
1 & RobotManipulation-PT-00001 & 110 & 274 \\ \hline 
2 & TokenRing-PT-005 & 166 & 365  \\ \hline 
3 & Philosophers-PT-000005 & 243 & 945  \\ \hline 
4 & RobotManipulation-PT-00002 & 1430 & 5 500  \\ \hline 
5 & Dekker-PT-010 & 6144 & 171 530  \\ \hline 
6 & BART-PT-002 & 17 424 & 53 328  \\ \hline 
7 & ClientsAndServers-PT-N0001P0 & 27 576 & 113 316  \\ \hline 
8 & Philosophers-PT-000010 & 59 049 & 459 270  \\ \hline 
9 & Referendum-PT-0010 & 59 050 & 393 661  \\ \hline 
10 & SatelliteMemory-PT-X00100Y0003 & 76 358 & 209 484  \\ \hline 
11 & RobotManipulation-PT-00005 & 184 756 & 1 137 708  \\ \hline 
12 & Dekker-PT-015 & 278 528 & 16 834 575\\ \hline 
13 & HouseConstruction-PT-00005 & 1 187 984 & 7 191 110 \\ \hline
\end{tabular}
\caption{Petri net models used for the experiments, for which we reported here the number of nodes and edges in the reachability graph (RG).}
\label{tab:models}
\end{table}


\begin{figure}[tbp]
    \centering
    \includegraphics[width=\textwidth]{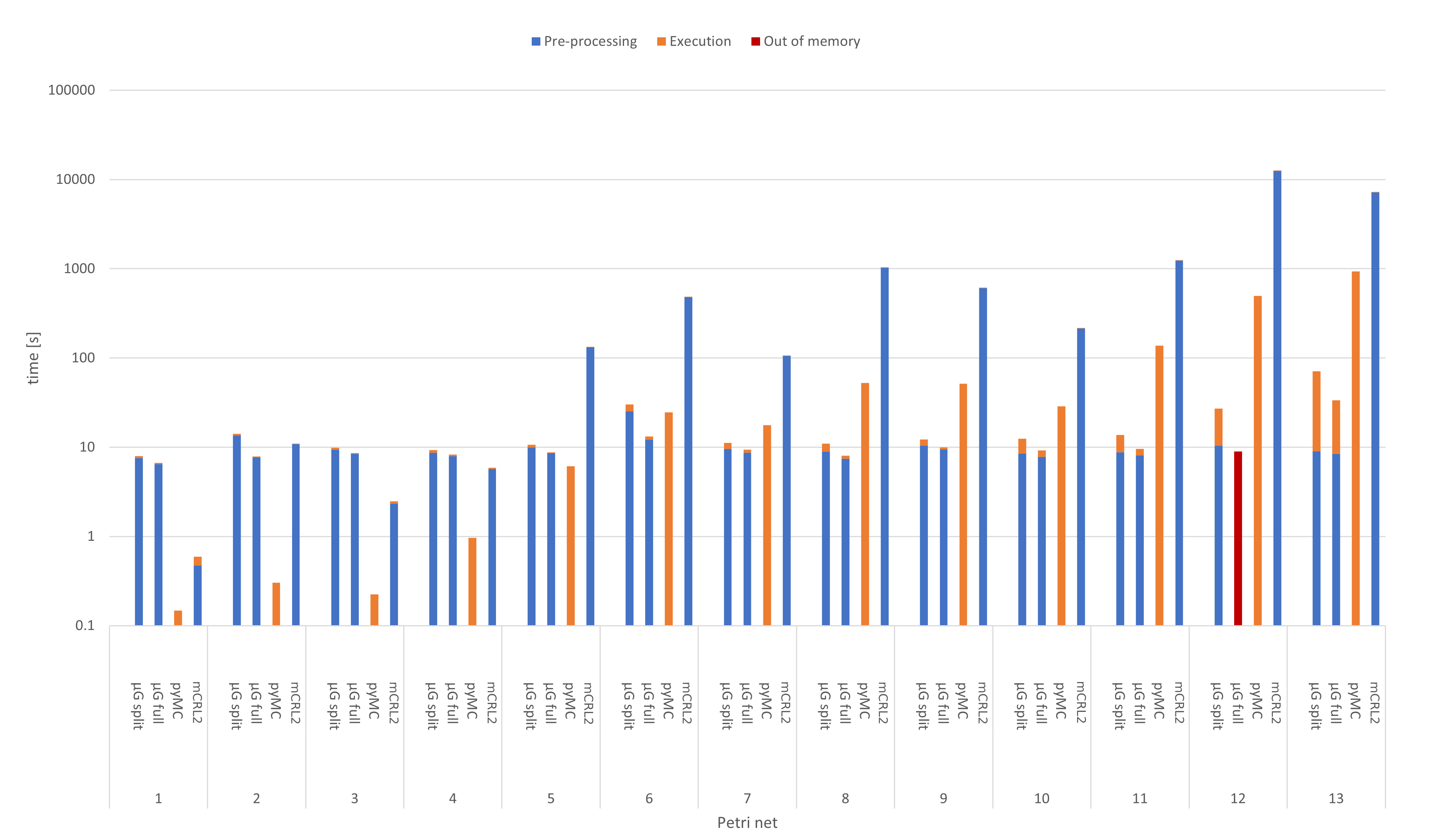}
    \caption{Execution times of $\mG$ (split and full setups), pyModelChecking, and mCRL2 across 13 Petri net models from the Model Checking Contest 2022. The y-axis is on a logarithmic scale. On the twelfth Petri net, the red bar indicates that the $\mG$ model ran out of memory after the pre-processing step.}
    \Description{A bar plot showing the execution times of $\mG$, pyModelChecking and mCRL2.}
    \label{fig:evaluation}
\end{figure}

\subsection{Specification of Graph Neural Networks}\label{sec:gnn-examples}
In this section, we show how to define some of the most popular graph neural network models in $\mG$. As is the case for any programming language, there are many different ways to obtain the same result, therefore the implementations that we show here are only intended to highlight the generality of the language. For specific use cases, it is up to the programmer to choose the most useful $\mG$ implementation.

\paragraph{Graph Convolutional Networks}
A Graph Convolutional Network~\cite{kipf_semi-supervised_2017} (GCN) performs a graph convolution by multiplying the node features with a matrix of weights, normalized according to the degree of the nodes and assigning greater importance to the features received from nodes with fewer neighbors. 
The new node labels $x_i' \in \mathbb{R}^m$ are computed from the current node labels $x_i \in \mathbb{R}^n$ according to the equation
\[x_i' = f \left(\sum_{j \in \fullset{\graph}{i} \cup \{i\}} \frac{x_j \Theta}{\sqrt{deg(i)}\sqrt{deg(j)}}\right)\]
where $deg(i)$ is the degree of node $i$, $\Theta \in \mathbb{R}^{n \times m}$ is a matrix of trainable weights, and $f$ is an activation function. The set of neighbors of a node may include both the predecessors and the successors, as the GCN was originally thought for the case of undirected graphs, or only the predecessors, as is usually the case for directed graphs. In this example, we will implement the GCN assuming the former definition of neighborhood.
\begin{table}[tbp]
    \centering
    \begin{tabular}{ll}
        $\psi_{NN}: (\mathbb{R}^{n})^\star \times \mathbb{R}^n \rightarrow \mathbb{R}^m$ & $\psi_{NN}(X, x) = f(x^T\Theta)$\\
        $\psi_{+_2}: (\mathbb{R}^{2})^\star \times \mathbb{R}^2 \rightarrow \mathbb{R}$ & $\psi_{+_2}(X, x) = x_1 + x_2$\\
        $\psi_{+_{2n}}: (\mathbb{R}^{2n})^\star \times \mathbb{R}^{2n} \rightarrow \mathbb{R}^n$ & $\psi_{+_{2n}}(X, x) = x_1 + x_2$\\
        $\varphi_{1}: \mathbb{R}^n \times \mathbb{R} \times \mathbb{R}^n \rightarrow \mathbb{R}$ & $\varphi_{1}(i, e, j) = 1$\\
        $\varphi_{dgn}: \mathbb{R}^{n+1} \times \mathbb{R} \times \mathbb{R}^{n+1} \rightarrow \mathbb{R}^n$ & $\varphi_{dgn}(i, e, j) = \frac{i_2}{\sqrt{i_1}\sqrt{j_1}}$\\
        $\sigma_{+}:  \mathbb{R}^\star \times \mathbb{R}^n \rightarrow \mathbb{R}$ & $\sigma_{+}(M, x) = 1 + \sum_{i}M_i$\\
        $\sigma_{dg+}:  (\mathbb{R}^{n})^\star \times \mathbb{R}^{n+1} \rightarrow \mathbb{R}^n$ & $\sigma_{dg+}(M, x) = \frac{x_2}{x_1} + \sum_{i}M_i$\\
    \end{tabular}
    \caption{Functions used in the definition of the GCN.}
    \label{tab:gcn}
\end{table}
To implement the GCN, we make use of the functions shown in Table~\ref{tab:gcn}. We start by defining a $\mG$ expression to compute the degree of each node. We recall that the degree of a node is the number of its outgoing edges (the \textit{outdegree}) plus the number of incoming edges (the \textit{indegree}), with self loops counting as both. It is easy to compute this number in $\mG$: we use $\lhd_{+}^{1}$ to compute the indegree and $\rhd_{+}^{1}$ to compute the outdegree, then we add them. 
\[(\lhd_{+}^{1} || \rhd_{+}^{1}) ; +_2\] 
In the steps that follow, we need the degree of nodes but also the original labels, therefore we use parallel composition with $\id$ to have both these values. 
\[(\lhd_{+}^{1} || \rhd_{+}^{1}) ; +_2 || \id\] 
Then we compute $\sum_{j \in \fullset{\graph}{i} \cup \{i\}} \frac{x_j}{\sqrt{deg(i)}\sqrt{deg(j)}}$ by summing the pre-image and the post-image using $dgn$ and $dg+$ 
\[((\lhd_{+}^{1} || \rhd_{+}^{1}) ; +_2 || \id) ; (\lhd_{dg+}^{dgn} || \rhd_{dg+}^{dgn}) ; +_{2n}\]
Note that both $\varphi_{dgn}$ and $\sigma_{dg+}$ expect node labels $l \in \mathbb{R}^{n+1}$, where $l_1 \in \mathbb{R}$ is the degree of the node and $l_2 \in \mathbb{R}^n$ is the actual label of the node. Finally, we multiply the node labels with the weight matrix $\Theta$ followed by the application of the activation function $f$, i.e., the classic dense neural network layer. Putting all of this together, the $\mG$ expression that implements a GCN layer is 
\begin{align*}
    &((\lhd_{+}^{1} || \rhd_{+}^{1}) ; +_2 || \id);(\lhd_{dg+}^{dgn} || \rhd_{dgr+}^{dgn}) ; +_{2n} ; NN
\end{align*}
\paragraph{Graph Attention Networks}
A Graph Attention Network (GAT)~\cite{velickovic_graph_2018}, like the GCN, transforms the node labels using a learnable weight matrix $\Theta \in \mathbb{R}^{n \times m}$. Then it employs a self-attention mechanism $a: \mathbb{R}^m \times \mathbb{R}^m \rightarrow \mathbb{R}$, shared by all the nodes, that given the labels of node $j$ and node $i$ such that $(j, i) \in \eSet$, it indicates the importance of node $j$'s labels to node $i$. This value is then normalized using the softmax function. In the original formulation, the attention mechanism is a single-layer feedforward neural network which uses a weight vector $\textbf{a} \in \mathbb{R}^{2m}$ and the LeakyReLU activation function. Therefore, the attention coefficients $\alpha_{ji}$ are given by
 \[\alpha_{ji} = \frac{exp \left( LeakyReLU \left( \textbf{a}^T (x_i \Theta, x_j \Theta) \right) \right)}{\sum_{k \in \preset{\graph}{i} \cup \{i\}} exp\left( LeakyReLU \left(\textbf{a}^T(x_i\Theta, x_k\Theta) \right)\right)}\]
Given these attention coefficients, the GAT layer computes the new node labels $x_i' \in \mathbb{R}^m$ from the current labels $x_i \in \mathbb{R}^n$ according to the equation
\[x_i' =  f\left(\sum_{j \in \preset{\graph}{i} \cup \{i\}} \alpha_{ji} x_j \Theta \right)\]
where $f$ is the activation function. The authors then extend this basic mechanism by employing multi-head attention, that is, they define $K$ independent attention heads and weight matrices whose outputs are then concatenated. Thus, we get
\[x_i' = \Big\Vert_{k=1}^K f\left(\sum_{j \in \preset{\graph}{i} \cup \{i\}} \alpha_{ji}^k x_j \Theta^k \right) \]
\begin{table}[tbp]
    \centering
    \begin{tabular}{ll}
    $\psi_{NN_i}: (\mathbb{R}^{n})^\star \times \mathbb{R}^n \rightarrow \mathbb{R}^m$ & $\psi_{NN_i}(X, x) = f(x\Theta^i)$\\
    $\psi_{/}: (\mathbb{R}^{n+1})^\star \times \mathbb{R}^{n+1} \rightarrow \mathbb{R}^n$ & $\psi_{/}(X, x) = \frac{x_1}{x_2}$\\
    $\varphi_{att*_i}: \mathbb{R}^n \times \mathbb{R} \times \mathbb{R}^n \rightarrow \mathbb{R}^n$ & $\varphi_{att*_i}(i, e, j) = exp(LeakyReLU(\textbf{a}_i \cdot (i^T\Theta^i , j^T\Theta^i))) \cdot i$\\
    $\varphi_{att_i}: \mathbb{R}^n \times \mathbb{R} \times \mathbb{R}^n \rightarrow \mathbb{R}$ & $\varphi_{att_i}(i, e, j) = exp(LeakyReLU(\textbf{a}_i \cdot (i^T\Theta^i , j^T\Theta^i)))$\\
    $\sigma_{t-att_i}:  \mathbb{R}^\star \times \mathbb{R}^n \rightarrow \mathbb{R}$ & $\sigma_{t-att_i}(M, x) = exp(LeakyReLU(\textbf{a}_i \cdot (x^T\Theta^i , x^T\Theta^i))) + \sum_{j}M_j$\\
    $\sigma_{n-att_i}:  (\mathbb{R}^{n})^\star \times \mathbb{R}^n \rightarrow \mathbb{R}^n$ & $\sigma_{n-att_i}(M, x) = exp(LeakyReLU(\textbf{a}_i \cdot (x^T\Theta^i , x^T\Theta^i))) \cdot x + \sum_{j}M_j$\\
    \end{tabular}
    \caption{Functions used in the definition of GAT convolutions.}
    \label{tab:gat}
\end{table}
The functions needed to implement the GAT layer are shown in Table~\ref{tab:gat}. First, we compute the unnormalized attention coefficients and multiply them with the node labels, using the term $\lhd_{n-att_1}^{att*_1}$. Then to normalize them, we compute in parallel the total attention values for incoming edges at each node using $\lhd_{t-att_1}^{att_1}$, and divide the labels of the first expression by this expression \[(\lhd_{n-att_1}^{att*_1} || \lhd_{t-att_1}^{att_1}) ; /\] At this point, we simply have to multiply by the weight matrix and apply the activation function to get the basic GAT layer \[(\lhd_{n-att_1}^{att*_1} || \lhd_{t-att_1}^{att_1}) ; / ; NN_1\]
Now we show how to use multi-head attention. We save the previous expression in $k$ different variables using \texttt{let} expressions, such that each expression gets its own attention mechanism and weight matrix. Then the body of the \texttt{let} expression is simply the parallel composition of these variables
\begin{align*}
    \texttt{let } &GAT_1 = (\lhd_{n-att_1}^{att*_1} || \lhd_{t-att_1}^{att_1}) ; / ; NN_1,\\
    &\qquad\qquad\qquad\qquad\vdots \\
    &GAT_k = (\lhd_{n-att_k}^{att*_k} || \lhd_{t-att_k}^{att_k}) ; / ; NN_k \texttt{ in } \\
    &\qquad\qquad GAT_1 || GAT_2 || \cdots || GAT_k
\end{align*}

\paragraph{Graph Isomorphism Networks}
A Graph Isomorphism Network (GIN)~\cite{xu_how_2019} generalizes the Weisfeiler-Lehman test~\cite{weisfeiler_reduction_1968} and thus has the greatest discriminative power among GNNs for determining whether two graphs are not isomorphic. The GIN convolution consists in multiplying each node label by $(1 + \epsilon)$, where $\epsilon \in \mathbb{R}$ is a trainable parameter, then adding the sum of the labels of the neighbor nodes. The obtained values are then passed in input to a multilayer perceptron (MLP), that is, it gets multiplied with a learnable weight matrix $\Theta^i$, followed by the application of an activation function $f_i$, two or more times in sequence. Formally, the GIN computes the new node labels $x_i' \in \mathbb{R}^m$ from the current labels $x_i \in \mathbb{R}^n$ as 
\[x_i' = MLP((1 + \epsilon) \cdot x_i + \sum_{j \in \preset{\graph}{i}}x_j)\]
\begin{table}[tbp]
    \centering
    \begin{tabular}{ll}
    $\psi_{MLP}: (\mathbb{R}^{n})^\star \times \mathbb{R}^n \rightarrow \mathbb{R}^{m}$ & $\psi_{MLP}(X, x) = f_2(f_1(x\Theta^{1})\Theta^2)$\\
    $\psi_{*\epsilon}: (\mathbb{R}^{n})^\star \times \mathbb{R}^n \rightarrow \mathbb{R}^n$ & $\psi_{*\epsilon}(X, x) = (w_\epsilon + 1) \cdot x$\\
    $\psi_{+}: (\mathbb{R}^{2n})^\star \times \mathbb{R}^{2n} \rightarrow \mathbb{R}^n$ & $\psi_{+}(X, x) = x_1 + x_2$\\
    $\varphi_{id}: \mathbb{R}^n \times \mathbb{R} \times \mathbb{R}^n \rightarrow \mathbb{R}^n$ & $\varphi_{*}(i, e, j) = i$\\
    $\sigma_{+}:  (\mathbb{R}^{n})^\star \times \mathbb{R}^n \rightarrow \mathbb{R}^n$ & $\sigma_{+}(M, x) = \sum_{i}M_i$\\
    \end{tabular}
    \caption{Functions used in the definition of GIN convolutions.}
    \label{tab:gin}
\end{table}
In Table~\ref{tab:gin} we show the functions we need to implement the GIN layer in $\mG$. We use parallel composition to compute the product of the node labels with $1 + \epsilon$ and the sum of the labels of the neighbors of each node, then we add them using $+$. 
\[(*\epsilon || \lhd^{id}_{+}) ; +\]
Finally, we put in sequence two dense neural network layers with weight matrices $\Theta^1 \in \mathbb{R}^{n \times n'}, \Theta^2 \in \mathbb{R}^{n' \times m}$ and activation functions $f_1, f_2$, as in the case of the original GIN article. Therefore, the $\mG$ expression that implements a GIN is 
\[(*\epsilon || \lhd^{id}_{+}) ; + ; MLP\]

\paragraph{The Original Graph Neural Network Model}
The graph neural network model was first introduced by~\citet{scarselli_graph_2009}. The new node labels $x_i' \in \mathbb{R}^m$ of the original GNN model are computed as a two-step process. In the first step the new ``state'' $s_i' \in \mathbb{R}^k$ of each node is computed from the current node labels $x_i \in \mathbb{R}^n$, the edge labels $e_{ji} \in \mathbb{R}^{l}$, and the neighbor nodes' labels $x_j \in \mathbb{R}^n$ and states $s_j \in \mathbb{R}^k$ as the fixed point of the equation 
\[s_i = \sum_{j \in \preset{\graph}{i}} h(x_i, e_{ji}, x_j, s_j)\] 
where $h$ is typically a multilayer perceptron with weight matrices $\Theta^1 \in \mathbb{R}^{(2n + l + k) \times n'}, \Theta^2 \in \mathbb{R}^{n' \times k}$ and activation functions $f_1, f_2$. Then, the second step consists in the application of a function $g$ to the state and labels of each node to obtain the final node labels 
\[x_i' = g(x_i, s_i)\]
where in this case too $g$ is a multilayer perceptron, with weight matrices $\Theta^3 \in \mathbb{R}^{(n + k) \times m'}, \Theta^4 \in \mathbb{R}^{m' \times m}$ and activation functions $f_3, f_4$.
\begin{table}[tbp]
    \centering
    \begin{tabular}{ll}
    $\psi_{MLP_2}: (\mathbb{R}^{n+k})^\star \times \mathbb{R}^{n+k} \rightarrow \mathbb{R}^{m}$ & $\psi_{MLP_2}(X, x) = f_4(f_3(x\Theta^{3})\Theta^{4})$\\
    $\psi_{0}: (\mathbb{R}^{n})^\star \times \mathbb{R}^n \rightarrow \mathbb{R}^k$ & $\psi_{0}(X, x) = \textbf{0}_k$\\
    $\varphi_{MLP_1}: \mathbb{R}^{n+k} \times \mathbb{R} \times \mathbb{R}^{n+k} \rightarrow \mathbb{R}^{k}$ & $\varphi_{MLP_1}(i, e, j) = f_2(f_1((\pi_L(j), e, i)\Theta^1)\Theta^2)$\\
    $\sigma_{+}:  (\mathbb{R}^{k})^\star \times \mathbb{R}^{k} \rightarrow \mathbb{R}^{k}$ & $\sigma_{+}(M, x) = \sum_{i}M_i$\\
    \end{tabular}
    \caption{Functions used in the definition of the original GNN convolution.}
    \label{tab:gnn}
\end{table}
In Table~\ref{tab:gnn} we show the functions required to implement this GNN layer. The first step consists in initializing the node states. For this example, we initialize the node states to an array of $k$ zeros and we append it to the current node labels using the identity term and parallel composition 
\[(\id || 0)\]
After this step, we compute the fixed point of the states using a MLP to generate messages and aggregating them through summation 
\[(\id || 0) ; \starop{(p_L || \lhd_{+}^{MLP_1})}\]
Once a fixed point for the states is reached, we use another MLP to obtain the final node labels and we get
\[(\id || 0);\starop{(p_L || \lhd_{+}^{MLP_1})};MLP_2\]

\section{Conclusion}\label{sec:conclusion}
We presented the syntax and semantics of $\mG$, a novel programming language for the definition of graph neural networks, where a graph neural network is characterized as a higher-order function $\gnnws: \lblfunvtype{T_1} \rightarrow \lblfunvtype{T_2}$ between node-labeling functions. We proved the type soundness of $\mG$ and shown a graphical formalism for the representation of $\mG$ programs. The language is available as a Python library built on the TensorFlow framework, and was evaluated on the application of CTL model checking and as a means to specify some of the most popular graph neural network models.

In future implementations, we plan to expand the language in a number of ways. A typical graph neural network operation that is still overlooked in $\mG$ is pooling. A pooling operator changes the underlying graph by reducing the number of nodes, producing a coarsened version of a graph. Graph pooling operators can be generally described in terms of the operations of selection, reduction, and connection~\cite{grattarola_understanding_2022}. The selection operator groups nodes together, the reduction operator aggregates the nodes in each group into a single node, and finally the connection operator generates the edges to link the newly generated nodes. The introduction of pooling operators requires careful consideration of how they can be composed with the existing terms, but it would allow $\mG$ to deal with graph-level tasks explicitly.

The main purpose for developing $\mG$ is to be eventually able to perform formal verification of graph neural networks and generate explanations for their outputs. For the time being, there are still no procedures that allow the verification of output reachability properties of GNNs~\cite{salzer_fundamental_2022}. In the future, we aim to define an abstract interpretation approach for these tasks.

\bibliographystyle{ACM-Reference-Format}
\bibliography{references}


\begin{thebibliography}{38}


\ifx \showCODEN    \undefined \def \showCODEN     #1{\unskip}     \fi
\ifx \showDOI      \undefined \def \showDOI       #1{#1}\fi
\ifx \showISBNx    \undefined \def \showISBNx     #1{\unskip}     \fi
\ifx \showISBNxiii \undefined \def \showISBNxiii  #1{\unskip}     \fi
\ifx \showISSN     \undefined \def \showISSN      #1{\unskip}     \fi
\ifx \showLCCN     \undefined \def \showLCCN      #1{\unskip}     \fi
\ifx \shownote     \undefined \def \shownote      #1{#1}          \fi
\ifx \showarticletitle \undefined \def \showarticletitle #1{#1}   \fi
\ifx \showURL      \undefined \def \showURL       {\relax}        \fi
\providecommand\bibfield[2]{#2}
\providecommand\bibinfo[2]{#2}
\providecommand\natexlab[1]{#1}
\providecommand\showeprint[2][]{arXiv:#2}

\bibitem[Abadi et~al\mbox{.}(2016)]%
        {tensorflow_2015}
\bibfield{author}{\bibinfo{person}{Mart\'{\i}n Abadi}, \bibinfo{person}{Paul Barham}, \bibinfo{person}{Jianmin Chen}, \bibinfo{person}{Zhifeng Chen}, \bibinfo{person}{Andy Davis}, \bibinfo{person}{Jeffrey Dean}, \bibinfo{person}{Matthieu Devin}, \bibinfo{person}{Sanjay Ghemawat}, \bibinfo{person}{Geoffrey Irving}, \bibinfo{person}{Michael Isard}, \bibinfo{person}{Manjunath Kudlur}, \bibinfo{person}{Josh Levenberg}, \bibinfo{person}{Rajat Monga}, \bibinfo{person}{Sherry Moore}, \bibinfo{person}{Derek~G. Murray}, \bibinfo{person}{Benoit Steiner}, \bibinfo{person}{Paul Tucker}, \bibinfo{person}{Vijay Vasudevan}, \bibinfo{person}{Pete Warden}, \bibinfo{person}{Martin Wicke}, \bibinfo{person}{Yuan Yu}, {and} \bibinfo{person}{Xiaoqiang Zheng}.} \bibinfo{year}{2016}\natexlab{}.
\newblock \showarticletitle{TensorFlow: a system for large-scale machine learning}. In \bibinfo{booktitle}{\emph{Proceedings of the 12th USENIX Conference on Operating Systems Design and Implementation}} (Savannah, GA, USA) \emph{(\bibinfo{series}{OSDI'16})}. \bibinfo{publisher}{USENIX Association}, \bibinfo{address}{USA}, \bibinfo{pages}{265–283}.
\newblock
\showISBNx{9781931971331}


\bibitem[Amparore et~al\mbox{.}(2019)]%
        {abcdg2019}
\bibfield{author}{\bibinfo{person}{Elvio Amparore}, \bibinfo{person}{Bernard Berthomieu}, \bibinfo{person}{Gianfranco Ciardo}, \bibinfo{person}{Silvano Dal~Zilio}, \bibinfo{person}{Francesco Gall{\`a}}, \bibinfo{person}{Lom~Messan Hillah}, \bibinfo{person}{Francis Hulin-Hubard}, \bibinfo{person}{Peter~Gj{\o}l Jensen}, \bibinfo{person}{Lo{\"i}g Jezequel}, \bibinfo{person}{Fabrice Kordon}, \bibinfo{person}{Didier Le~Botlan}, \bibinfo{person}{Torsten Liebke}, \bibinfo{person}{Jeroen Meijer}, \bibinfo{person}{Andrew Miner}, \bibinfo{person}{Emmanuel Paviot-Adet}, \bibinfo{person}{Ji{\v{r}}{\'i} Srba}, \bibinfo{person}{Yann Thierry-Mieg}, \bibinfo{person}{Tom van Dijk}, {and} \bibinfo{person}{Karsten Wolf}.} \bibinfo{year}{2019}\natexlab{}.
\newblock \showarticletitle{Presentation of the 9th Edition of the Model Checking Contest}. In \bibinfo{booktitle}{\emph{Tools and Algorithms for the Construction and Analysis of Systems}}, \bibfield{editor}{\bibinfo{person}{Dirk Beyer}, \bibinfo{person}{Marieke Huisman}, \bibinfo{person}{Fabrice Kordon}, {and} \bibinfo{person}{Bernhard Steffen}} (Eds.). \bibinfo{publisher}{Springer International Publishing}, \bibinfo{address}{Cham}, \bibinfo{pages}{50--68}.
\newblock
\showISBNx{978-3-030-17502-3}
\urldef\tempurl%
\url{https://doi.org/10.1007/978-3-030-17502-3_4}
\showDOI{\tempurl}


\bibitem[Athalye et~al\mbox{.}(2018)]%
        {athalye_synthesizing_2018}
\bibfield{author}{\bibinfo{person}{Anish Athalye}, \bibinfo{person}{Logan Engstrom}, \bibinfo{person}{Andrew Ilyas}, {and} \bibinfo{person}{Kevin Kwok}.} \bibinfo{year}{2018}\natexlab{}.
\newblock \bibinfo{title}{Synthesizing Robust Adversarial Examples}.
\newblock
\newblock
\showeprint[arxiv]{1707.07397}~[cs.CV]


\bibitem[Barbierato and Gatti(2024)]%
        {barbierato_challenges_2024}
\bibfield{author}{\bibinfo{person}{Enrico Barbierato} {and} \bibinfo{person}{A. Gatti}.} \bibinfo{year}{2024}\natexlab{}.
\newblock \showarticletitle{The Challenges of Machine Learning: A Critical Review}.
\newblock \bibinfo{journal}{\emph{ELECTRONICS}}  \bibinfo{volume}{13} (\bibinfo{year}{2024}), \bibinfo{pages}{1--30}.
\newblock
\showISSN{2079-9292}
\urldef\tempurl%
\url{https://doi.org/10.3390/electronics13020416}
\showDOI{\tempurl}


\bibitem[Battaglia et~al\mbox{.}(2018)]%
        {battaglia_relational_2018}
\bibfield{author}{\bibinfo{person}{Peter~W. Battaglia}, \bibinfo{person}{Jessica~B. Hamrick}, \bibinfo{person}{Victor Bapst}, \bibinfo{person}{Alvaro Sanchez-Gonzalez}, \bibinfo{person}{Vinicius Zambaldi}, \bibinfo{person}{Mateusz Malinowski}, \bibinfo{person}{Andrea Tacchetti}, \bibinfo{person}{David Raposo}, \bibinfo{person}{Adam Santoro}, \bibinfo{person}{Ryan Faulkner}, \bibinfo{person}{Caglar Gulcehre}, \bibinfo{person}{Francis Song}, \bibinfo{person}{Andrew Ballard}, \bibinfo{person}{Justin Gilmer}, \bibinfo{person}{George Dahl}, \bibinfo{person}{Ashish Vaswani}, \bibinfo{person}{Kelsey Allen}, \bibinfo{person}{Charles Nash}, \bibinfo{person}{Victoria Langston}, \bibinfo{person}{Chris Dyer}, \bibinfo{person}{Nicolas Heess}, \bibinfo{person}{Daan Wierstra}, \bibinfo{person}{Pushmeet Kohli}, \bibinfo{person}{Matt Botvinick}, \bibinfo{person}{Oriol Vinyals}, \bibinfo{person}{Yujia Li}, {and} \bibinfo{person}{Razvan Pascanu}.} \bibinfo{year}{2018}\natexlab{}.
\newblock \bibinfo{title}{Relational inductive biases, deep learning, and graph networks}.
\newblock
\newblock
\urldef\tempurl%
\url{https://doi.org/10.48550/arXiv.1806.01261}
\showDOI{\tempurl}
\newblock
\shownote{arXiv:1806.01261 [cs, stat]}.


\bibitem[Belenchia et~al\mbox{.}(2023)]%
        {belenchia_implementing_2023}
\bibfield{author}{\bibinfo{person}{Matteo Belenchia}, \bibinfo{person}{Flavio Corradini}, \bibinfo{person}{Michela Quadrini}, {and} \bibinfo{person}{Michele Loreti}.} \bibinfo{year}{2023}\natexlab{}.
\newblock \showarticletitle{Implementing a {CTL} {Model} {Checker} with {$\mu \mathcal{G}$}, a {Language} for {Programming} {Graph} {Neural} {Networks}}. In \bibinfo{booktitle}{\emph{Formal {Techniques} for {Distributed} {Objects}, {Components}, and {Systems}}} \emph{(\bibinfo{series}{Lecture {Notes} in {Computer} {Science}})}, \bibfield{editor}{\bibinfo{person}{Marieke Huisman} {and} \bibinfo{person}{António Ravara}} (Eds.). \bibinfo{publisher}{Springer Nature Switzerland}, \bibinfo{address}{Cham}, \bibinfo{pages}{37--54}.
\newblock
\showISBNx{978-3-031-35355-0}
\urldef\tempurl%
\url{https://doi.org/10.1007/978-3-031-35355-0_4}
\showDOI{\tempurl}


\bibitem[Belenchia et~al\mbox{.}(2024)]%
        {belenchia_libmg_2024}
\bibfield{author}{\bibinfo{person}{Matteo Belenchia}, \bibinfo{person}{Flavio Corradini}, \bibinfo{person}{Michela Quadrini}, {and} \bibinfo{person}{Michele Loreti}.} \bibinfo{year}{2024}\natexlab{}.
\newblock \showarticletitle{{\textsc{libmg}}: A Python library for programming graph neural networks in {$\mu \mathcal{G}$}}.
\newblock \bibinfo{journal}{\emph{Science of Computer Programming}}  \bibinfo{volume}{238} (\bibinfo{year}{2024}), \bibinfo{pages}{103165}.
\newblock
\showISSN{0167-6423}
\urldef\tempurl%
\url{https://doi.org/10.1016/j.scico.2024.103165}
\showDOI{\tempurl}


\bibitem[Bronstein et~al\mbox{.}(2021)]%
        {bronstein_geometric_2021}
\bibfield{author}{\bibinfo{person}{Michael~M. Bronstein}, \bibinfo{person}{Joan Bruna}, \bibinfo{person}{Taco Cohen}, {and} \bibinfo{person}{Petar Veličković}.} \bibinfo{year}{2021}\natexlab{}.
\newblock \bibinfo{title}{Geometric Deep Learning: Grids, Groups, Graphs, Geodesics, and Gauges}.
\newblock
\newblock
\showeprint[arxiv]{2104.13478}~[cs.LG]


\bibitem[Bunte et~al\mbox{.}(2019)]%
        {mucrl2}
\bibfield{author}{\bibinfo{person}{Olav Bunte}, \bibinfo{person}{Jan~Friso Groote}, \bibinfo{person}{Jeroen J.~A. Keiren}, \bibinfo{person}{Maurice Laveaux}, \bibinfo{person}{Thomas Neele}, \bibinfo{person}{Erik~P. de Vink}, \bibinfo{person}{Wieger Wesselink}, \bibinfo{person}{Anton Wijs}, {and} \bibinfo{person}{Tim A.~C. Willemse}.} \bibinfo{year}{2019}\natexlab{}.
\newblock \showarticletitle{The mCRL2 Toolset for Analysing Concurrent Systems}. In \bibinfo{booktitle}{\emph{Tools and Algorithms for the Construction and Analysis of Systems}}, \bibfield{editor}{\bibinfo{person}{Tom{\'a}{\v{s}} Vojnar} {and} \bibinfo{person}{Lijun Zhang}} (Eds.). \bibinfo{publisher}{Springer International Publishing}, \bibinfo{address}{Cham}, \bibinfo{pages}{21--39}.
\newblock
\showISBNx{978-3-030-17465-1}
\urldef\tempurl%
\url{https://doi.org/10.1007/978-3-030-17465-1_2}
\showDOI{\tempurl}


\bibitem[Corso et~al\mbox{.}(2024)]%
        {corso_graph_2024}
\bibfield{author}{\bibinfo{person}{Gabriele Corso}, \bibinfo{person}{Hannes Stark}, \bibinfo{person}{Stefanie Jegelka}, \bibinfo{person}{Tommi Jaakkola}, {and} \bibinfo{person}{Regina Barzilay}.} \bibinfo{year}{2024}\natexlab{}.
\newblock \showarticletitle{Graph neural networks}.
\newblock \bibinfo{journal}{\emph{Nature Reviews Methods Primers}} \bibinfo{volume}{4}, \bibinfo{number}{1} (\bibinfo{date}{07 Mar} \bibinfo{year}{2024}), \bibinfo{pages}{17}.
\newblock
\showISSN{2662-8449}
\urldef\tempurl%
\url{https://doi.org/10.1038/s43586-024-00294-7}
\showDOI{\tempurl}


\bibitem[Elango et~al\mbox{.}(2018)]%
        {elango_diesel_2018}
\bibfield{author}{\bibinfo{person}{Venmugil Elango}, \bibinfo{person}{Norm Rubin}, \bibinfo{person}{Mahesh Ravishankar}, \bibinfo{person}{Hariharan Sandanagobalane}, {and} \bibinfo{person}{Vinod Grover}.} \bibinfo{year}{2018}\natexlab{}.
\newblock \showarticletitle{Diesel: {DSL} for linear algebra and neural net computations on {GPUs}}. In \bibinfo{booktitle}{\emph{Proceedings of the 2nd {ACM} {SIGPLAN} {International} {Workshop} on {Machine} {Learning} and {Programming} {Languages}}} \emph{(\bibinfo{series}{{MAPL} 2018})}. \bibinfo{publisher}{Association for Computing Machinery}, \bibinfo{address}{New York, NY, USA}, \bibinfo{pages}{42--51}.
\newblock
\showISBNx{978-1-4503-5834-7}
\urldef\tempurl%
\url{https://doi.org/10.1145/3211346.3211354}
\showDOI{\tempurl}


\bibitem[Eykholt et~al\mbox{.}(2018)]%
        {eykholt_robust_2018}
\bibfield{author}{\bibinfo{person}{Kevin Eykholt}, \bibinfo{person}{Ivan Evtimov}, \bibinfo{person}{Earlence Fernandes}, \bibinfo{person}{Bo Li}, \bibinfo{person}{Amir Rahmati}, \bibinfo{person}{Chaowei Xiao}, \bibinfo{person}{Atul Prakash}, \bibinfo{person}{Tadayoshi Kohno}, {and} \bibinfo{person}{Dawn Song}.} \bibinfo{year}{2018}\natexlab{}.
\newblock \bibinfo{title}{Robust Physical-World Attacks on Deep Learning Models}.
\newblock
\newblock
\showeprint[arxiv]{1707.08945}~[cs.CR]


\bibitem[García~Díaz et~al\mbox{.}(2015)]%
        {garcia_diaz_towards_2015}
\bibfield{author}{\bibinfo{person}{Vicente García~Díaz}, \bibinfo{person}{Jordán Espada}, \bibinfo{person}{B. Pelayo García-Bustelo}, {and} \bibinfo{person}{Juan Cueva~Lovelle}.} \bibinfo{year}{2015}\natexlab{}.
\newblock \showarticletitle{Towards a {Standard}-based {Domain}-specific {Platform} to {Solve} {Machine} {Learning}-based {Problems}}.
\newblock \bibinfo{journal}{\emph{International Journal of Artificial Intelligence and Interactive Multimedia}}  \bibinfo{volume}{3} (\bibinfo{date}{Jan.} \bibinfo{year}{2015}), \bibinfo{pages}{6--12}.
\newblock
\urldef\tempurl%
\url{https://doi.org/10.9781/ijimai.2015.351}
\showDOI{\tempurl}


\bibitem[Gilmer et~al\mbox{.}(2017)]%
        {gilmer_neural_2017}
\bibfield{author}{\bibinfo{person}{Justin Gilmer}, \bibinfo{person}{Samuel~S. Schoenholz}, \bibinfo{person}{Patrick~F. Riley}, \bibinfo{person}{Oriol Vinyals}, {and} \bibinfo{person}{George~E. Dahl}.} \bibinfo{year}{2017}\natexlab{}.
\newblock \showarticletitle{Neural message passing for {Quantum} chemistry}. In \bibinfo{booktitle}{\emph{Proceedings of the 34th {International} {Conference} on {Machine} {Learning} - {Volume} 70}} \emph{(\bibinfo{series}{{ICML}'17})}. \bibinfo{publisher}{JMLR.org}, \bibinfo{address}{Sydney, NSW, Australia}, \bibinfo{pages}{1263--1272}.
\newblock


\bibitem[Grattarola and Alippi(2020)]%
        {grattarola_graph_2021}
\bibfield{author}{\bibinfo{person}{Daniele Grattarola} {and} \bibinfo{person}{Cesare Alippi}.} \bibinfo{year}{2020}\natexlab{}.
\newblock \showarticletitle{Graph Neural Networks in TensorFlow and Keras with Spektral}.
\newblock \bibinfo{journal}{\emph{IEEE Comput. Intell. Mag.}}  \bibinfo{volume}{16} (\bibinfo{year}{2020}), \bibinfo{pages}{99--106}.
\newblock
\urldef\tempurl%
\url{https://doi.org/10.1109/mci.2020.3039072}
\showDOI{\tempurl}


\bibitem[Grattarola et~al\mbox{.}(2022)]%
        {grattarola_understanding_2022}
\bibfield{author}{\bibinfo{person}{Daniele Grattarola}, \bibinfo{person}{Daniele Zambon}, \bibinfo{person}{Filippo~Maria Bianchi}, {and} \bibinfo{person}{Cesare Alippi}.} \bibinfo{year}{2022}\natexlab{}.
\newblock \showarticletitle{Understanding {Pooling} in {Graph} {Neural} {Networks}}.
\newblock \bibinfo{journal}{\emph{IEEE Transactions on Neural Networks and Learning Systems}} \bibinfo{volume}{35}, \bibinfo{number}{2} (\bibinfo{date}{July} \bibinfo{year}{2022}), \bibinfo{pages}{2708--2718}.
\newblock
\showISSN{2162-2388}
\urldef\tempurl%
\url{https://doi.org/10.1109/TNNLS.2022.3190922}
\showDOI{\tempurl}
\newblock
\shownote{Conference Name: IEEE Transactions on Neural Networks and Learning Systems}.


\bibitem[Han et~al\mbox{.}(2023)]%
        {sicong_interpreting_2023}
\bibfield{author}{\bibinfo{person}{Sicong Han}, \bibinfo{person}{Chenhao Lin}, \bibinfo{person}{Chao Shen}, \bibinfo{person}{Qian Wang}, {and} \bibinfo{person}{Xiaohong Guan}.} \bibinfo{year}{2023}\natexlab{}.
\newblock \showarticletitle{Interpreting Adversarial Examples in Deep Learning: A Review}.
\newblock \bibinfo{journal}{\emph{ACM Comput. Surv.}} \bibinfo{volume}{55}, \bibinfo{number}{14s}, Article \bibinfo{articleno}{328} (\bibinfo{date}{jul} \bibinfo{year}{2023}), \bibinfo{numpages}{38}~pages.
\newblock
\showISSN{0360-0300}
\urldef\tempurl%
\url{https://doi.org/10.1145/3594869}
\showDOI{\tempurl}


\bibitem[Huang et~al\mbox{.}(2020)]%
        {huang_survey_2020}
\bibfield{author}{\bibinfo{person}{Xiaowei Huang}, \bibinfo{person}{Daniel Kroening}, \bibinfo{person}{Wenjie Ruan}, \bibinfo{person}{James Sharp}, \bibinfo{person}{Youcheng Sun}, \bibinfo{person}{Emese Thamo}, \bibinfo{person}{Min Wu}, {and} \bibinfo{person}{Xinping Yi}.} \bibinfo{year}{2020}\natexlab{}.
\newblock \bibinfo{title}{A {Survey} of {Safety} and {Trustworthiness} of {Deep} {Neural} {Networks}: {Verification}, {Testing}, {Adversarial} {Attack} and {Defence}, and {Interpretability}}.
\newblock
\newblock
\urldef\tempurl%
\url{http://arxiv.org/abs/1812.08342}
\showURL{%
\tempurl}
\newblock
\shownote{arXiv:1812.08342 [cs]}.


\bibitem[Jahić et~al\mbox{.}(2023)]%
        {jahic_semkis-dsl_2023}
\bibfield{author}{\bibinfo{person}{Benjamin Jahić}, \bibinfo{person}{Nicolas Guelfi}, {and} \bibinfo{person}{Benoît Ries}.} \bibinfo{year}{2023}\natexlab{}.
\newblock \showarticletitle{{SEMKIS}-{DSL}: {A} {Domain}-{Specific} {Language} to {Support} {Requirements} {Engineering} of {Datasets} and {Neural} {Network} {Recognition}}.
\newblock \bibinfo{journal}{\emph{Information}} \bibinfo{volume}{14}, \bibinfo{number}{4} (\bibinfo{date}{April} \bibinfo{year}{2023}), \bibinfo{pages}{213}.
\newblock
\showISSN{2078-2489}
\urldef\tempurl%
\url{https://doi.org/10.3390/info14040213}
\showDOI{\tempurl}
\newblock
\shownote{Number: 4 Publisher: Multidisciplinary Digital Publishing Institute}.


\bibitem[Kearns(1998)]%
        {kearns_efficient_1998}
\bibfield{author}{\bibinfo{person}{Michael Kearns}.} \bibinfo{year}{1998}\natexlab{}.
\newblock \showarticletitle{Efficient noise-tolerant learning from statistical queries}.
\newblock \bibinfo{journal}{\emph{J. ACM}} \bibinfo{volume}{45}, \bibinfo{number}{6} (\bibinfo{date}{Nov.} \bibinfo{year}{1998}), \bibinfo{pages}{983--1006}.
\newblock
\showISSN{0004-5411}
\urldef\tempurl%
\url{https://doi.org/10.1145/293347.293351}
\showDOI{\tempurl}


\bibitem[Kipf and Welling(2017)]%
        {kipf_semi-supervised_2017}
\bibfield{author}{\bibinfo{person}{Thomas~N. Kipf} {and} \bibinfo{person}{Max Welling}.} \bibinfo{year}{2017}\natexlab{}.
\newblock \bibinfo{title}{Semi-Supervised Classification with Graph Convolutional Networks}.
\newblock
\newblock
\showeprint[arxiv]{1609.02907}~[cs.LG]


\bibitem[Lamb et~al\mbox{.}(2020)]%
        {lamb_graph_2020}
\bibfield{author}{\bibinfo{person}{Luís~C. Lamb}, \bibinfo{person}{Artur~d’Avila Garcez}, \bibinfo{person}{Marco Gori}, \bibinfo{person}{Marcelo~O.R. Prates}, \bibinfo{person}{Pedro~H.C. Avelar}, {and} \bibinfo{person}{Moshe~Y. Vardi}.} \bibinfo{year}{2020}\natexlab{}.
\newblock \showarticletitle{Graph Neural Networks Meet Neural-Symbolic Computing: A Survey and Perspective}. In \bibinfo{booktitle}{\emph{Proceedings of the Twenty-Ninth International Joint Conference on Artificial Intelligence, {IJCAI-20}}}, \bibfield{editor}{\bibinfo{person}{Christian Bessiere}} (Ed.). \bibinfo{publisher}{International Joint Conferences on Artificial Intelligence Organization}, \bibinfo{address}{Yokohama, Yokohama, Japan}, \bibinfo{pages}{4877--4884}.
\newblock
\urldef\tempurl%
\url{https://doi.org/10.24963/ijcai.2020/679}
\showDOI{\tempurl}
\newblock
\shownote{Survey track}.


\bibitem[Marcus(2018a)]%
        {marcus_deep_2019}
\bibfield{author}{\bibinfo{person}{Gary Marcus}.} \bibinfo{year}{2018}\natexlab{a}.
\newblock \bibinfo{title}{Deep Learning: {A} Critical Appraisal}.
\newblock
\newblock
\showeprint[arXiv]{1801.00631}
\urldef\tempurl%
\url{http://arxiv.org/abs/1801.00631}
\showURL{%
\tempurl}


\bibitem[Marcus(2018b)]%
        {marcus_innateness_2018}
\bibfield{author}{\bibinfo{person}{Gary Marcus}.} \bibinfo{year}{2018}\natexlab{b}.
\newblock \bibinfo{title}{Innateness, AlphaZero, and Artificial Intelligence}.
\newblock
\newblock
\showeprint[arxiv]{1801.05667}~[cs.AI]


\bibitem[Marcus(2020)]%
        {marcus_next_2020}
\bibfield{author}{\bibinfo{person}{Gary Marcus}.} \bibinfo{year}{2020}\natexlab{}.
\newblock \bibinfo{title}{The {Next} {Decade} in {AI}: {Four} {Steps} {Towards} {Robust} {Artificial} {Intelligence}}.
\newblock
\newblock
\urldef\tempurl%
\url{https://arxiv.org/abs/2002.06177v3}
\showURL{%
\tempurl}
\newblock
\shownote{arXiv: 2002.06177}.


\bibitem[Morris et~al\mbox{.}(2021)]%
        {morris_power_2021}
\bibfield{author}{\bibinfo{person}{Christopher Morris}, \bibinfo{person}{Matthias Fey}, {and} \bibinfo{person}{Nils Kriege}.} \bibinfo{year}{2021}\natexlab{}.
\newblock \showarticletitle{The {Power} of the {Weisfeiler}-{Leman} {Algorithm} for {Machine} {Learning} with {Graphs}}. In \bibinfo{booktitle}{\emph{Proceedings of the {Thirtieth} {International} {Joint} {Conference} on {Artificial} {Intelligence}}}. \bibinfo{publisher}{ijcai.org}, \bibinfo{address}{Montreal, Canada}, \bibinfo{pages}{4543--4550}.
\newblock
\urldef\tempurl%
\url{https://doi.org/10.24963/ijcai.2021/618}
\showDOI{\tempurl}
\newblock
\shownote{Conference Name: Thirtieth International Joint Conference on Artificial Intelligence \{IJCAI-21\} ISBN: 9780999241196 Place: Montreal, Canada Publisher: International Joint Conferences on Artificial Intelligence Organization}.


\bibitem[Podobas et~al\mbox{.}(2021)]%
        {podobas_streambrain_2021}
\bibfield{author}{\bibinfo{person}{Artur Podobas}, \bibinfo{person}{Martin Svedin}, \bibinfo{person}{Steven W.~D. Chien}, \bibinfo{person}{Ivy~B. Peng}, \bibinfo{person}{Naresh~Balaji Ravichandran}, \bibinfo{person}{Pawel Herman}, \bibinfo{person}{Anders Lansner}, {and} \bibinfo{person}{Stefano Markidis}.} \bibinfo{year}{2021}\natexlab{}.
\newblock \showarticletitle{StreamBrain: An HPC Framework for Brain-like Neural Networks on CPUs, GPUs and FPGAs}. In \bibinfo{booktitle}{\emph{Proceedings of the 11th International Symposium on Highly Efficient Accelerators and Reconfigurable Technologies}} (Online, Germany) \emph{(\bibinfo{series}{HEART '21})}. \bibinfo{publisher}{Association for Computing Machinery}, \bibinfo{address}{New York, NY, USA}, Article \bibinfo{articleno}{8}, \bibinfo{numpages}{6}~pages.
\newblock
\showISBNx{9781450385497}
\urldef\tempurl%
\url{https://doi.org/10.1145/3468044.3468052}
\showDOI{\tempurl}


\bibitem[Portugal et~al\mbox{.}(2016)]%
        {portugal_preliminary_2016}
\bibfield{author}{\bibinfo{person}{Ivens Portugal}, \bibinfo{person}{Paulo Alencar}, {and} \bibinfo{person}{Donald Cowan}.} \bibinfo{year}{2016}\natexlab{}.
\newblock \showarticletitle{A {Preliminary} {Survey} on {Domain}-{Specific} {Languages} for {Machine} {Learning} in {Big} {Data}}. In \bibinfo{booktitle}{\emph{2016 {IEEE} {International} {Conference} on {Software} {Science}, {Technology} and {Engineering} ({SWSTE})}}. \bibinfo{publisher}{IEEE}, \bibinfo{address}{Beer Sheva, Israel}, \bibinfo{pages}{108--110}.
\newblock
\urldef\tempurl%
\url{https://doi.org/10.1109/SWSTE.2016.23}
\showDOI{\tempurl}


\bibitem[Sarker(2021)]%
        {sarker_deep_2021}
\bibfield{author}{\bibinfo{person}{Iqbal~H. Sarker}.} \bibinfo{year}{2021}\natexlab{}.
\newblock \showarticletitle{Deep Learning: A Comprehensive Overview on Techniques, Taxonomy, Applications and Research Directions}.
\newblock \bibinfo{journal}{\emph{SN Computer Science}} \bibinfo{volume}{2}, \bibinfo{number}{6} (\bibinfo{date}{18 Aug} \bibinfo{year}{2021}), \bibinfo{pages}{420}.
\newblock
\showISSN{2661-8907}
\urldef\tempurl%
\url{https://doi.org/10.1007/s42979-021-00815-1}
\showDOI{\tempurl}


\bibitem[Scarselli et~al\mbox{.}(2009)]%
        {scarselli_graph_2009}
\bibfield{author}{\bibinfo{person}{Franco Scarselli}, \bibinfo{person}{Marco Gori}, \bibinfo{person}{Ah~Chung Tsoi}, \bibinfo{person}{Markus Hagenbuchner}, {and} \bibinfo{person}{Gabriele Monfardini}.} \bibinfo{year}{2009}\natexlab{}.
\newblock \showarticletitle{The {Graph} {Neural} {Network} {Model}}.
\newblock \bibinfo{journal}{\emph{IEEE Transactions on Neural Networks}} \bibinfo{volume}{20}, \bibinfo{number}{1} (\bibinfo{date}{Jan.} \bibinfo{year}{2009}), \bibinfo{pages}{61--80}.
\newblock
\showISSN{1941-0093}
\urldef\tempurl%
\url{https://doi.org/10.1109/TNN.2008.2005605}
\showDOI{\tempurl}
\newblock
\shownote{Conference Name: IEEE Transactions on Neural Networks}.


\bibitem[Sujeeth et~al\mbox{.}(2011)]%
        {sujeeth_optiml_2011}
\bibfield{author}{\bibinfo{person}{Arvind~K. Sujeeth}, \bibinfo{person}{HyoukJoong Lee}, \bibinfo{person}{Kevin~J. Brown}, \bibinfo{person}{Hassan Chafi}, \bibinfo{person}{Michael Wu}, \bibinfo{person}{Anand~R. Atreya}, \bibinfo{person}{Kunle Olukotun}, \bibinfo{person}{Tiark Rompf}, {and} \bibinfo{person}{Martin Odersky}.} \bibinfo{year}{2011}\natexlab{}.
\newblock \showarticletitle{{OptiML}: an implicitly parallel domain-specific language for machine learning}. In \bibinfo{booktitle}{\emph{Proceedings of the 28th {International} {Conference} on {International} {Conference} on {Machine} {Learning}}} \emph{(\bibinfo{series}{{ICML}'11})}. \bibinfo{publisher}{Omnipress}, \bibinfo{address}{Madison, WI, USA}, \bibinfo{pages}{609--616}.
\newblock
\showISBNx{978-1-4503-0619-5}


\bibitem[Susskind et~al\mbox{.}(2021)]%
        {susskind_neuro-symbolic_2021}
\bibfield{author}{\bibinfo{person}{Zachary Susskind}, \bibinfo{person}{Bryce Arden}, \bibinfo{person}{Lizy~K. John}, \bibinfo{person}{Patrick Stockton}, {and} \bibinfo{person}{Eugene~B. John}.} \bibinfo{year}{2021}\natexlab{}.
\newblock \bibinfo{title}{Neuro-Symbolic {AI:} An Emerging Class of {AI} Workloads and their Characterization}.
\newblock
\newblock
\showeprint[arXiv]{2109.06133}
\urldef\tempurl%
\url{https://arxiv.org/abs/2109.06133}
\showURL{%
\tempurl}


\bibitem[Sälzer and Lange(2022)]%
        {salzer_fundamental_2022}
\bibfield{author}{\bibinfo{person}{Marco Sälzer} {and} \bibinfo{person}{Martin Lange}.} \bibinfo{year}{2022}\natexlab{}.
\newblock \bibinfo{title}{Fundamental {Limits} in {Formal} {Verification} of {Message}-{Passing} {Neural} {Networks}}.
\newblock
\newblock
\urldef\tempurl%
\url{https://openreview.net/forum?id=WlbG820mRH-}
\showURL{%
\tempurl}


\bibitem[Veličković et~al\mbox{.}(2018)]%
        {velickovic_graph_2018}
\bibfield{author}{\bibinfo{person}{Petar Veličković}, \bibinfo{person}{Guillem Cucurull}, \bibinfo{person}{Arantxa Casanova}, \bibinfo{person}{Adriana Romero}, \bibinfo{person}{Pietro Liò}, {and} \bibinfo{person}{Yoshua Bengio}.} \bibinfo{year}{2018}\natexlab{}.
\newblock \bibinfo{title}{Graph Attention Networks}.
\newblock
\newblock
\showeprint[arxiv]{1710.10903}~[stat.ML]


\bibitem[Weisfeiler and Leman(1968)]%
        {weisfeiler_reduction_1968}
\bibfield{author}{\bibinfo{person}{Boris~Yulievich Weisfeiler} {and} \bibinfo{person}{Andrey~Aleksandrovich Leman}.} \bibinfo{year}{1968}\natexlab{}.
\newblock \showarticletitle{The {Reduction} of a {Graph} to {Canonical} {Form} and the {Algebra} which {Appears} therein}.
\newblock \bibinfo{journal}{\emph{Nauchno-Technicheskaya Informatsia}}  \bibinfo{volume}{9} (\bibinfo{year}{1968}), \bibinfo{pages}{12 -- 16}.
\newblock


\bibitem[Xu et~al\mbox{.}(2019)]%
        {xu_how_2019}
\bibfield{author}{\bibinfo{person}{Keyulu Xu}, \bibinfo{person}{Weihua Hu}, \bibinfo{person}{Jure Leskovec}, {and} \bibinfo{person}{Stefanie Jegelka}.} \bibinfo{year}{2019}\natexlab{}.
\newblock \bibinfo{title}{How Powerful are Graph Neural Networks?}
\newblock
\newblock
\showeprint[arxiv]{1810.00826}~[cs.LG]


\bibitem[Zhao et~al\mbox{.}(2017)]%
        {zhao_deepdsl_2017}
\bibfield{author}{\bibinfo{person}{Tian Zhao}, \bibinfo{person}{Xiaobing Huang}, {and} \bibinfo{person}{Yu Cao}.} \bibinfo{year}{2017}\natexlab{}.
\newblock \bibinfo{title}{DeepDSL: A Compilation-based Domain-Specific Language for Deep Learning}.
\newblock
\newblock
\showeprint[arxiv]{1701.02284}~[cs.PL]


\bibitem[Zhou et~al\mbox{.}(2020)]%
        {zhou_graph_2020}
\bibfield{author}{\bibinfo{person}{Jie Zhou}, \bibinfo{person}{Ganqu Cui}, \bibinfo{person}{Shengding Hu}, \bibinfo{person}{Zhengyan Zhang}, \bibinfo{person}{Cheng Yang}, \bibinfo{person}{Zhiyuan Liu}, \bibinfo{person}{Lifeng Wang}, \bibinfo{person}{Changcheng Li}, {and} \bibinfo{person}{Maosong Sun}.} \bibinfo{year}{2020}\natexlab{}.
\newblock \showarticletitle{Graph neural networks: {A} review of methods and applications}.
\newblock \bibinfo{journal}{\emph{AI Open}}  \bibinfo{volume}{1} (\bibinfo{date}{Jan.} \bibinfo{year}{2020}), \bibinfo{pages}{57--81}.
\newblock
\showISSN{2666-6510}
\urldef\tempurl%
\url{https://doi.org/10.1016/j.aiopen.2021.01.001}
\showDOI{\tempurl}


\end{thebibliography}


\end{document}